\documentclass[11pt]{article}

\usepackage{soul}
\usepackage{url}
\usepackage[hidelinks]{hyperref}
\usepackage[utf8]{inputenc}
\usepackage[small]{caption}
\usepackage{graphicx}
\usepackage{booktabs}

\usepackage{amsmath,multirow}

\usepackage{amsthm,thm-restate,thmtools}
\usepackage{amssymb}

\usepackage{enumerate}
\usepackage{etoolbox}
\usepackage{subcaption}
\usepackage{xcolor,xfrac}
\usepackage[sort,numbers]{natbib}
\usepackage{hyperref}
\usepackage[english]{babel}
\usepackage{graphicx,charter}
\usepackage{authblk}
\usepackage{framed}
\usepackage[normalem]{ulem}
\usepackage{amsmath,multirow}
\usepackage{amsthm,thm-restate,thmtools}
\usepackage{tikz}
\usetikzlibrary{shapes.arrows}
\usetikzlibrary{shapes.geometric}
\usepackage{amssymb}
\usepackage{amsfonts}
\usepackage{enumerate}
\usepackage{etoolbox}
\usepackage[T1]{fontenc}
\usepackage[utf8]{inputenc}
\usepackage[top=1 in,bottom=1in, left=1 in, right=1 in]{geometry}
\usepackage{xcolor,xfrac}
\usepackage{url}
\usepackage{enumitem}
\usepackage{tabularx}
\usepackage{algorithm}
\usepackage{algorithmic}

\usepackage{xcolor}
\usepackage{thm-restate}
\usepackage{hyperref}
\hypersetup{
     colorlinks   = true,
     linkcolor    = red, 
     urlcolor     = teal, 
	 citecolor    = blue 
}
\usepackage{mathtools}

\DeclareMathOperator*{\argmax}{arg\,max}
\DeclareMathOperator*{\argmin}{arg\,min}

\renewcommand{\emptyset}{\varnothing}

\newtheorem{theorem}{Theorem}[section]
\newtheorem{lemma}{Lemma}[section]
\newtheorem{proposition}{Proposition}[section]

\newtheorem{definition}{Definition}

\newtheorem{example}{Example}

\newcommand{\EFone}{\textrm{\textup{EF1}}}
\newcommand{\MMS}{\textrm{\textup{MMS}}}
\newcommand{\EF}{\textrm{\textup{EF}}}
\newcommand{\EFX}{\textrm{\textup{EFX}}}
\newcommand{\PO}{\textup{PO}}
\newcommand{\SO}{\textup{SO}}
\usepackage{pifont}
\newcommand{\cmark}{\ding{51}}%
\newcommand{\xmark}{\ding{55}}%
\usepackage{booktabs}
\usepackage{nicefrac}
\newcommand{\sort}{\mathrm{sort}}

\newcommand{\NPH}{\textrm{\textup{NP-hard}}}
\newcommand{\Poly}{\textup{P}}
\newcommand{\dist}{\mathrm{dist}}

\newcommand{\SN}[1]{{\color{blue!50!magenta}{SN: }{#1} }}

\usepackage[capitalise,noabbrev]{cleveref}
\Crefname{remark}{Remark}{Remarks}
\Crefname{rmk}{Remark}{Remarks}
\Crefname{dfn}{Definition}{Definitions}
\Crefname{thm}{Theorem}{Theorems}
\Crefname{cor}{Corollary}{Corollaries}
\Crefname{lem}{Lemma}{Lemmas}
\Crefname{example}{Example}{Examples}
\Crefname{prop}{Proposition}{Propositions}
\Crefname{tab}{Table}{Tables}

\title{Fair Distribution of Delivery Orders}

\author[1]{Hadi Hosseini}
\author[2]{Shivika Narang}
\author[3]{Tomasz Wąs}
\affil[1]{Pennsylvania State University}
\affil[2]{University of New South Wales}
\affil[3]{University of Oxford}
\date{}

\begin{document}
\maketitle
\begin{abstract}
\noindent We initiate the study of fair distribution of delivery tasks among a set of agents wherein delivery jobs are placed along the vertices of a graph.
Our goal is to fairly distribute delivery costs (distance traveled to complete the deliveries) among a fixed set of agents while satisfying some desirable notions of economic efficiency. We adopt well-established fairness concepts---such as \textit{envy-freeness up to one item} (\EFone{}) and \textit{minimax share} (\MMS{})---to our setting and show that fairness is often incompatible with the efficiency notion of \textit{social optimality}. We then characterize instances that admit fair and socially optimal solutions by exploiting graph structures. We further show that achieving fairness along with Pareto optimality is computationally intractable. We complement this by designing an XP algorithm (parameterized by the number of agents) for finding \MMS{} and Pareto optimal solutions on every tree instance, and show that the same algorithm can be modified to find efficient solutions along with \EFone{}, when such solutions exist.
The latter crucially relies on an intriguing result that in our setting \EFone{} and Pareto optimality jointly imply \MMS{}.
We conclude by theoretically and experimentally analyzing the price of fairness.
\end{abstract}

\section{Introduction}\label{sec:intro}

With the rise of digital marketplaces and the gig economy, package delivery services have become crucial components of e-commerce platforms like Amazon, AliExpress, and eBay. In addition to these novel platforms, traditional postal and courier services also  require swift turnarounds for distributing packages. 
Prior work has extensively investigated the optimal partitioning of tasks among the delivery agents under the guise of \textit{vehicle routing} problems (see \citep{toth2002overview} for an overview). However, these solutions are primarily focused on optimizing the efficiency (often measured by delivery time or distance travelled \citep{kleinberg1999fairness,pioro2007fair}),
and do not consider fairness towards the delivery agents.
This is particularly important in settings where agents do not receive monetary compensation,
e.g., in volunteer-based social programs such as Meals on Wheels~\citep{odwyer2009mealsonwheels}.

We consider fair distribution of delivery orders that are located on the vertices of a connected graph, containing a warehouse (the \textit{hub}). Agents are tasked with picking up delivery packages (or \textit{items}) from the fixed hub, delivering them to the vertices, and returning to the hub. In this setting, the cost incurred by an agent $i$ is the total distance traveled, that is, the total number of the edges traversed by $i$ in the graph.
Let us illustrate this through an example.

\begin{example}\label{ex:example}
\begin{figure}[!t]
    \centering
    \begin{tikzpicture}
        \def\xs{1cm} 
        \def\ys{0.6cm} 
        \def\x{0cm} 
        \def\y{0cm} 
        \def\ls{\footnotesize} 
        
        \tikzset{
            node_blank/.style={circle,draw,minimum size=0.5cm,inner sep=0, color=white}, 
            node/.style={circle,draw,minimum size=0.6cm,inner sep=0, fill = black!05},
            node_h/.style={circle,draw,minimum size=0.8cm,inner sep=0, fill = blue!20, font=\normalsize},
            node_emph/.style={circle, minimum size=1cm, black!15, fill = black!15, draw,inner sep=0, font=\normalsize},
            edge/.style={draw,thick,sloped,-,above,font=\footnotesize},
            arrow/.style={draw, single arrow, minimum width = 0.9cm, minimum height=\y-6*\x+\s, fill=black!10},
            blank/.style={}
        }
        
        \node[node_emph] (_) at (\x + 1*\xs, 1*\ys + \y) {};
        \node[node] (a) at (\x + -0.1*\xs, 1*\ys + \y) {\ls $a$};
        \node[node_h] (h) at (\x + 1*\xs, 1*\ys + \y) {\ls $h$};
        \node[node] (b) at (\x + 2.1*\xs, 1*\ys + \y) {\ls $b$};
        \node[node] (c) at (\x + 3*\xs, 0*\ys + \y) {\ls $c$};
        \node[node] (d) at (\x + 3*\xs, 2*\ys + \y) {\ls $d$};
        \node[node] (e) at (\x + 4*\xs, 2*\ys + \y) {\ls $e$};
        \node[node] (f) at (\x + 5*\xs, 2*\ys + \y) {\ls $f$};
        \node[node] (g) at (\x + 6*\xs, 2*\ys + \y) {\ls $g$};
        
        \path[edge]
        (a) edge (h)
        (h) edge (b)
        (b) edge (c)
        (b) edge (d)
        (d) edge (e)
        (e) edge (f)
        (f) edge (g)
        ;
      
    \end{tikzpicture}
    \caption{An example graph with the hub, $h$, marked.}
    \label{fig:motiv}
\end{figure}
Consider seven delivery orders $\{a, b, \ldots, g\}$ and a hub ($h$) that are located on a graph as depicted in \cref{fig:motiv}.
An agent's cost depends on the graph structure and is \emph{submodular}. For instance, the cost of delivering an order to vertex $f$ is the distance from the hub $h$ to $f$, which is $4$%
;\footnote{Formally, there is also the cost of returning to $h$, but since, on trees, each edge must be traversed by an agent twice (once in each direction), we do not count the return cost for simplicity.
}%
 but the cost of delivering to $f$ and $g$ is only $5$ since they can both be serviced in the same trip.

Let there be two (delivery) agents. If the objective were to simply minimize the total distance travelled (\emph{social optimality}), then there are two solutions with the total cost of $7$: either one agent delivers \emph{all} the items or one agent services $a$ while the other services the rest.
However, these solutions do not distribute the delivery orders fairly among the  agents.

One plausible fair solution may assign $\{a,b,f\}$ to the first agent and $\{c,d,e,g\}$ to the other, minimizing the cost discrepancy.
However, as both agents benefit from exchanging $f$ for $c$,
this allocation is not efficient
or, more precisely, it is not \emph{Pareto optimal}.
After the exchange, the first agent services $\{a,b,c\}$
and the second agent $\{d,e,f,g\}$,
which in fact is a Pareto optimal allocation.
\end{example}

The above example captures the challenges in satisfying fairness in conjunction with efficiency, and consequently, motivates the study of fair distribution of delivery orders. The literature on fair division has long been concerned with the fair allocation of goods (or resources) \citep{lipton-envy-graph,budish2011combinatorial,barman2019fair,freeman2019equitable}, chores (or tasks) \citep{chen2020fairness,ebadian2021fairly,huang2021algorithmic,hosseini2022ordinalgoods}, and mixtures thereof \citep{bhaskar2021approximate,aziz2022fair,caragiannis2022repeatedly,hosseini2022fairly}. 

It has resulted in a variety of fairness concepts and their relaxations. Most notably, \textit{envy-freeness} and its relaxation---\textit{envy-freeness up to one item} (\EFone{}) \citep{lipton-envy-graph}---have been widely studied in the context of fair division.
Another well-studied fairness notion, minimax share (\MMS{}) \citep{budish2011combinatorial}, requires that agents receive cost no more than what they would have received if they were to create (almost) equal partitions.
A key question is how to adopt these fairness concepts to the delivery problems, and whether these fairness concepts are compatible with natural efficiency requirements.


\subsection{Technical Contributions}
We initiate the study of fair distribution of delivery tasks among a set of agents. The tasks are placed on the vertices of an acyclic or tree graph. The cost of servicing a given set of tasks is the number of edges traversed to be able to reach each node.  This cost function can be easily seen to be submodular. 
The primary objective is to find a fair partition of $m$ delivery orders (represented by vertices of a graph), starting from a fixed hub, among $n$ agents. 
We consider two well-established fairness concepts of \EFone{} and \MMS{} and explore their existence and computation along with efficiency notions of social optimality (\SO{}) and Pareto optimality (\PO{}).

\begin{table}[t!]
\centering
\setlength{\tabcolsep}{2pt}
\small
\begin{tabular}{@{}lllll@{}}
\toprule
         & & \multicolumn{1}{l}{   \textbf{--}}
         & \multicolumn{1}{l}{\textbf{\PO{}}}
         & \multicolumn{1}{l}{\textbf{\SO{}}} \\
\midrule

\multirow{4}{*}{\textbf{\EFone{}}} &
     \multirow{2}{*}{existence}   &
        \cmark{} &
        \xmark{}  &
        \xmark{}  \\
        &&{\footnotesize (\cref{prop:ef1})} &{\footnotesize (\cref{prop:non-existenceEFonePO}a)} &{\footnotesize (\cref{prop:non-existenceEFonePO}a)}\\
        & \multirow{2}{*}{computation} &
        \Poly{} &
        NP-h  &
        NP-h \\ 
        &&{\footnotesize (\cref{prop:ef1})} & {\footnotesize (\cref{prop:po:hardness})} &{\footnotesize (\cref{prop:soexist})} \\
        \midrule

\multirow{4}{*}{\textbf{\MMS{}}} &
   \multirow{2}{*}{existence}   &
        \cmark{} & 
        \cmark{} &
        \xmark{}  \\
        &&{\footnotesize (\cref{thm:mms:hardness})}&{\footnotesize (\cref{prop:non-existenceEFonePO}b)} &{\footnotesize (\cref{thm:so:characterization})} \\
    & \multirow{2}{*}{computation} & 
        NP-h  & 
        NP-h &
        NP-h \\ 
        &&{\footnotesize (\cref{thm:mms:hardness})}&{\footnotesize (\cref{prop:po:hardness})}&{\footnotesize (\cref{prop:soexist})}\\
        \bottomrule  
\end{tabular}
\caption{The summary of our results on \EFone{} and \MMS{} in conjunction with efficiency requirements of PO and SO. 
\cmark{} denotes that the allocation always exists, and \xmark{} that it may not exist. 
We note that every \NPH{} problem here can be solved with an XP algorithm parameterized by the number of agents (Theorem~\ref{thm:xpalg:alloc}).
}
\label{tab:contribution}
\end{table}

\paragraph{Existence of fair and efficient allocations.} \cref{tab:contribution} summarizes our results on the coexistence of fairness and efficiency in arbitrary graphs.
We first show that an \EFone{} allocation always exists and can be computed in polynomial time on trees (\cref{prop:ef1}). In contrast, an \MMS{} allocation is guaranteed to exist, but its computation remains \NPH{} (\cref{thm:mms:hardness}). Out of the various combinations of fairness and efficiency, only \MMS{} and \PO{} allocations are guaranteed to exist. Finding a fair and efficient allocation remains \NPH{} for each combination.


We note that our intractability results hold even for the restricted case of unweighted tree graphs.
Consequently, in this paper we focus on the unweighted tree instances only.
As we show, even this special case allows for a rich landscape of results.
By doing a comprehensive study of the fair delivery problem on trees,
we aim to establish a solid baseline for further work in the more general settings.
Indeed, a popular technique to deal with computationally difficult problems on graphs
is to provide algorithms with complexity parameterized by some measure of tree-resemblance~[See for instance Chapter 7 in \citet{cygan2015parameterized}].
Trees can be also motivated in practice. Many suburban neighborhoods in the United States, among other countries, are designed with a cul-de-sac layout, with houses/mailboxes typically located at uniform distances. Then, a group of agents distributing a newspaper or a leaflet to every house in such a neighborhood, will face an unweighted tree instance, as in our problem.
However, we note that many of our results already hold for the general case of cyclic and weighted graph, which we discuss in \cref{sec:extend}.

In \cref{sec:charac},
we characterize the conditions for
the existence of 
fair (\EFone{} and \MMS{}) and efficient (\PO{} and \SO{}) allocations
in tree instances.
In particular, one of our most technically involved results provides a necessary condition for an allocation to be \EFone{}  and \PO{} (\cref{prop:non-existenceEFonePO}). 
In turn, this helps us to prove that such an allocation must be leximin optimal. 
As a result, we show that an \EFone{} and \PO{} allocation {\em always satisfies} \MMS{} as well (\cref{thrm:ef1+po:implies_mms}). 

In \cref{sec:efx}, we present results on a stronger fairness notion of envy-freeness up to any item (\EFX{}). We note that an \EFX{} allocation always exists in this setting ,
however \EFX{} combined with some efficiency requirement,
for large classes of graphs,
becomes as restrictive
as envy-freeness(\cref{thrm:efx}). Further determining whether an efficient \EFX{} allocation exists is also \NPH{} (\cref{prop:efx+po:hardness}).

\paragraph{Exact Algorithm. } We complement our existence results by
designing an XP algorithm (\cref{alg:ef1+po:main}),
parameterized by the number of agents,
that finds the Pareto frontier of a given delivery instance (\cref{thm:xpalg}). 
This allows us to find an \MMS{} and \PO{} solution. Further, our characterization results in conjunction with \cref{alg:ef1+po:main} enable us to check whether an instance admits an \EFone{} and \PO{} allocation or a fair and \SO{} allocation as well (\cref{thm:xpalg:alloc}). This algorithm also forms enables us to do extensive experiments on randomly generated instances. 

\paragraph{Price of Fairness and Experiments.} In our last set of theoretical results, we study the price of fairness of \MMS{} and \EFone{}
in our setting. Informally, price of fairness of a given instance is the ratio of the minimum possible sum of agent costs under a fair allocation to the minimum possible sum of agent costs under any allocation. To this end, we discuss the efficiency and computational complexity of finding minimum cost fair allocations \cref{prop:mincostexist,thm:mincostfair}. 

We formally study price of fairness by establishing  worst case upper bounds and typical values in a theoretical and experimental analysis, respectively. We present the theoretical bounds on the price of fairness in \cref{sec:pof} and the experimental findings on price of fairness along with other experiments in \cref{sec:exp}.  Our other experiments include checking how often \EFone{} and \PO{} allocations and fair and \SO{} allocations exist and studying the trajectories of the Pareto frontiers of randomly generated instances.


\subsection{Related Work}
Fair division of indivisible items has garnered much attention in recent years. Several notions of fairness have been explored in this space, with \EFone{} \citep{lipton-envy-graph,budish2011combinatorial,barman2019fair,caragiannis2022repeatedly} and \MMS{} \citep{ghodsi2018fair,barman2020approximation,hosseini2021guaranteeing} being among the most prominent ones. An important result is from \citet{caragiannis2019unreasonable} showing that an \EFone{} and \PO{} allocation is guaranteed to exist for items with non-negative additive valuations. Some prior work has also looked at fair division on graphs \citep{bouveret2017fair,truszczynski2020maximin,misra2021equitable,bilo2022almost,misra2022fair},
but in the settings that are very different from assigning delivery orders.
The majority of these papers identified the vertices of a graph with goods and analyzed how to fairly partition the graph into contiguous pieces.

Some recent work has explored fairness in delivery settings \citep{gupta2022fairfoody,nair2022gigs,wu2022fair,singh2023fairassign,tsai2023genetic} or ride-hailing platforms
\citep{esmaeili2022rawlsian,sanchez2022balancing}. However, in these studies tasks cannot be combined to give lower aggregate cost than the sum of the individual costs. This is very different from our setting, where an agent delivering one order can deliver all orders on the way for no additional cost. Further, the prior work on these settings is largely experimental and does not provide any positive theoretical guarantees.  
While fairness has been studied in routing problems, the aim has been to balance the amount of traffic on  each edge \citep{kleinberg1999fairness,pioro2007fair}, which does not capture the type of delivery instances that we investigate in this paper. 
In \cref{app:priorwork}, we provide an extended review of the literature. 



\section{Our Model}
\label{sec:prelim}
We denote a \emph{delivery instance} by an ordered triple $I = \langle [n],G,h \rangle$, where $[n]$ is a set of agents, $G = (V,E)$ is an undirected acyclic graph (i.e., a tree) of delivery orders rooted in $h \in V$. The special vertex, i.e., the root, $h$ is called the \emph{hub}. By $m$, we denote the number of edges in the graph, which is also the number of non-hub vertices. We assume that $m\geq n$.

The goal is to assign each vertex in graph $G$, except for the hub, to a unique agent that will \emph{service} it. Formally, an allocation $A=(A_1,\dots,A_n)$ is an $n$-partition of vertices in $V \setminus \{h\}$. We are only interested in \textit{complete} allocations such that $\cup_{i\in [n]} A_i = V \setminus \{h\}$, and denote the set of all complete allocations by $\Pi^n$.

An agent's \emph{cost} for servicing a vertex $v\in V$, denoted by $c(v)$, is the length of the shortest path from the hub $h$ to $v$.
An agent's cost for servicing a set of vertices $S \subseteq V \setminus \{h\}$ is equal to the minimum length of a walk that starts and ends in $h$ and contains all vertices in $S$ divided by two.\footnote{On trees, in each such walk, each edge is traversed by an agent two times (once in both directions). For simplicity, we drop the return cost, hence the division by 2.}
A walk to service vertices in $S$ may pass through vertices in some superset of $S$, i.e., $S' \supseteq S$.
Thus, the cost function $c$ is 
\textit{submodular} and belongs to the class of  \textit{coverage functions}.
We use $G|_S$ to denote the minimal connected subgraph containing all vertices in $S \cup \{h\}$. Thus, we have $c(S) = |E(G|_S)|$.
We say that an agent servicing $S$, \textit{visits} all vertices in $G|_S$.

\paragraph{Fairness Concepts.}
The most plausible fairness notion is \textit{envy-freeness} (\EF{}), which requires that no agent (strictly) prefers the allocation of another agent. An \EF{} allocation may not exist; consider one delivery order and two agents.
A prominent relaxation of \EF{} is \textit{envy-freeness up to one order} (\EFone) \citep{lipton-envy-graph,budish2011combinatorial},
which requires that every pairwise envy can be eliminated by the removal of a single order served by the envious agent.

\begin{definition}[Envy-Freeness up to One Order (\EFone)]
An allocation $A$ is \EFone{} if for every pair $i,j\in [n]$,
either $A_i = \emptyset$ or
there exists $x\in A_i$
such that $c(A_i\setminus \{x\})\leq c(A_j)$.
\end{definition}

Another well-studied notion is \textit{minimax share} (\MMS{}), which ensures that each agent gets at most as much cost as they would if they were to create an $n$-partition of the delivery orders but then receive their least preferred bundle.
This notion is an adaptation of maximin share fairness---which was defined for  positive valuations \citep{budish2011combinatorial}---to settings with negative valuations and has been recently studied in fair allocation of chores (see e.g., \citet{huang2021algorithmic}).

\begin{definition}[Minimax Share (\MMS{})]
The minimax share cost of a given delivery instance $I$ is given by
\[
\MMS{} (I) = \min_{A\in \Pi^n}\max_{i\in [n]} c(A_i),
\]
An allocation $A$ is \MMS{} if $c(A_i) \!\leq\! \MMS{}_i(I)$ for all $i\!\in\! [n]$.
\end{definition}

Under additive non-negative identical valuations, \MMS{} and \EFone{} co-exist.
For completeness, we give a formal proof in \cref{sec:mmsandfriends}.
However, as our cost functions are submodular, neither \EFone{} nor \MMS{} implies the other, even though the cost functions are identical
(see \cref{ex:allocations}).

\paragraph{Economic Efficiency.}
Our first notion of efficiency is social optimality (\SO{}). \SO{} requires that the aggregate cost of all of the agents is minimum. That is, we want to minimize the number of agents that traverse any given edge. One simple way to do this would be to have a single agent traverse the entire graph. Here, the sum of all agents' costs is $m$ i.e., equal to the number of edges in the graph. Clearly this is the minimum possible cost of any allocation. Thus, an allocation is socially optimal, if and only if, the sum of all agents' costs is $m$.
Equivalently, \SO{} is satisfied, if and only if, each vertex (except for the hub) is visited by exactly one agent.

\begin{definition}[Social Optimality (\SO{})]
An allocation $A$ is \emph{socially optimal} if $\sum_{i=1}^n c(A_i) = |E(G)|$.
In other words, for every pair of agents $i\neq j \in [n]$, the only vertex they both visit is the hub, i.e., $V(G|_{A_i}) \cap V(G|_{A_j})=\{h\}$.
\end{definition}

Here, $V(\cdot)$ takes as an input a subgraph and returns the set of vertices in it. An allocation that assigns all vertices to a single agent is vacuously \SO{}. However, as we discussed in \cref{ex:example} it may result in a very unfair distribution of orders. Therefore, we consider a weaker efficiency notion that allows for
some overlap in vertices visited by the agents.

\begin{definition}[Pareto Optimality (\PO{})]
An allocation $A$ \emph{Pareto dominates} $A'$ if $c(A_i)\leq c(A'_i)$, for every agent $i\in [n]$, and there exists some agent $j\in [n]$ such that $c(A_j)< c(A'_j)$.
An allocation is \emph{Pareto optimal} if it is not Pareto dominated by any other allocation.
\end{definition}
In other words, an allocation is \PO{}
if we cannot reduce the cost of one agent
without increasing it for some other agent.
Let us now follow up on \cref{ex:example}
and analyze allocations satisfying our notions.

\begin{example}[Continuation of \cref{ex:example}]
\label{ex:allocations}
    Consider the instance with
    2 agents and the graph from \cref{fig:motiv}.
    As previously noted,
    there are only two \SO{} allocations
    and neither is \EFone{} or \MMS{}.

    \PO{} allocation $(\{d,e,f,g\},\{a,b,c\})$ satisfies \MMS{}
    (vertex $g$ must be serviced by some agent, hence the \MMS{} cost cannot be smaller than 5),
    but it is not \EFone{}.
    In fact, there is no \EFone{} and \PO{} allocation in this instance,
    as an agent servicing $g$ has to service $f,e$ and $d$ as well
    (otherwise giving them to this agent would be a Pareto improvement).
    But then, even when we assign the remaining vertices, $a,b,c$, to the second agent,
    the allocation would violate \EFone{}.
    
    Finally, observe that allocation $(\{a,b,f\},\{c,d,e,g\})$
    is \EFone{}, but not \MMS{}.
\end{example} 

In order to study the (co-)existence of these fairness and efficiency notions, we use leximin optimality. Let us first set up some notation before defining it. Given an allocation $A$, we can \textit{sort it in non-increasing cost order}
to obtain allocation $B = \sort(A)$ such that
$c(B_1) \ge c(B_2) \ge \dots \ge c(B_n)$
and $B_i = A_{\pi(i)}$ for every agent $i \in [n]$ and some permutation of agents $\pi$.

\begin{definition}[Leximin Optimality]
    An allocation $A$ leximin dominates an allocation $A'$ if there is agent $i \in [n]$
    such that $c(B_i) < c(B'_i)$ and $c(B_j) = c(B'_j)$ for every $j \in [i-1]$,
    where $B=\sort(A)$ and $B'=\sort(A')$.
    An allocation is \textit{leximin optimal} if it is not leximin dominated by any other allocation.
\end{definition}

In other words, a leximin optimal allocation first minimizes the cost of the worst-off agent, then minimizes the cost of the second worst-off agent, and so on.
We note that a leximin optimal allocation is always Pareto optimal as well.


\section{Existence of Fair Allocations} \label{sec:fairallocations}

In this section, we consider the existence and computation of \EFone{} and \MMS{} allocations.

Recall that our cost functions are submodular and \textit{monotone}.
Thus, an \EFone{} allocation can be obtained by adapting an \textit{envy-graph} algorithm from \citet{lipton-envy-graph}: start with each agent having an empty set, pick an agent who currently has minimum cost (i.e., a sink of the envy-graph) and out of the unassigned vertices give them the one that results in a minimal increase of the cost\footnote{The original algorithm by \citet{lipton-envy-graph} also had an envy-cycle elimination step, which we do not need. Observe that as agents have identical cost functions, envy-cycles cannot arise here.}. Repeat this till all vertices in $V\setminus \{h\}$ are allocated.
This algorithm always returns an \EFone{} solution
as long as valuations are identical and monotone.

\begin{proposition}\label{prop:ef1}
Given a delivery instance $I=\langle [n],G, h\rangle$, an \EFone{} allocation always exists and can be computed in polynomial time.
\end{proposition}

We now shift our focus to allocations that satisfy \MMS{}. An \MMS{} allocation in our setting always exists. This follows from the fact
that agents have identical cost functions (an allocation that minimizes the maximum cost will satisfy \MMS{}).
However, finding such an allocation is \NPH{}.
To establish this, we first show the hardness
of finding the \MMS{} cost.

\begin{theorem}\label{thm:mms:hardness}
Given a delivery instance $I=\langle [n],G ,h\rangle$, an \MMS{} allocation always exists. However, finding  i) the \MMS{} cost in \NPH{} and ii) finding  an MMS allocation is \NPH{}.
\end{theorem}
\vspace{-3mm}
\begin{proof}
The existence of \MMS{} allocations follows from the cost functions being identical. The same allocation will give the \MMS{} threshold for all agents, hence would satisfy \MMS{} for all.

For part i) we give a reduction from \textsc{3-Partition} to a setting with unweighted trees.
In the \textsc{3-Partition} problem, we are given $3k$ positive integers
$S=\{s_1,\ldots, s_{3k}\}$ that sum up to $kT$ for some $T \in \mathbb{N}$.
The task is to decide if there is a partition of $S$
into $k$ pairwise disjoint subsets, $P = P_1,\dots,P_k \subseteq S$ such that
the elements in each subset sum up to $T$.
This problem is known to be \NPH{}~\citep{johnson1979computers}, even when the values of the integers are polynomial in $k$.

For each instance of \textsc{3-Partition} let us construct
a delivery instance with $k$ agents.
To this end, we construct a tree where for each integer $s_i \in S$,
let us take $s_i$ vertices $v^1_i,\dots,v^{s_i}_i$,
which, with the hub, $h$, gives as a total of $3kT+1$ vertices and all edges having a weight of $1$.
Next, for every $i \in [3k]$, let us connect all consecutive vertices in sequence $h, v^1_i,\dots,v^{s_i}_i$ with an edge to form a path.
In this way, we obtain a graph that consists of the hub and
$3k$ paths of different lengths outgoing from the hub
(such graphs are known as \textit{spider} or \textit{starlike} graphs).
See \cref{fig:mms:hardness} for an illustration.
As per our convention for trees, when defining edge costs, we count the traversal of each edge exactly once.

Let us show that minimax share cost in this instance is $T$,
if and only if,
there exists a desired partition in the original \textsc{3-Partition} instance.
If there is a partition $P$, consider allocation $A$ obtained by
assigning to every agent $j \in [k]$,
all vertices corresponding to integers in $P_j$,
i.e., $A_j = \cup_{s_i \in P_j} \{v^1_i,\dots,v^{s_i}_i\}$.
In this allocation,
the cost of every agent will be equal to $T$.
Further, the maximum cost of an agent in any allocation
cannot be smaller than $T$.
If not, it would make the total cost smaller than the number of edges in the graph, which is not possible.
Hence, the \MMS{} cost is $T$.

Conversely, 
if we know that the \MMS{} cost is $T$, we can show that a 3-Partition exists.
Take an arbitrary \MMS{} allocation $A$.
Consider the leaves in the bundle of an arbitrary agent $j \in [k]$,
i.e., vertices of the form $v^{s_i}_i \in A_j$ for some $i \in [3k]$.
Observe that to service each such leaf,
agent $j$ has to traverse $s_i$ edges
and these costs are summed when the agent services multiple leaves.
Now,
we know that the total cost of $j$ is at most $T$, i.e., $c(A_j)\le T$.
Thus, the sum of integers corresponding to the leaves serviced by each agents is at most $T$.
Hence, the leaves serviced by agents
give us a desired partition $P$.
Consequently, finding the \MMS{} cost is \NPH{}

For part ii), observe that in the constructed instance from part i), 
in an \MMS{} allocation $A$ the cost of the agent  that is the worst off, must be equal to \MMS{} cost, i.e.,
$\max_{i \in [n]} c(A_i) = \MMS{}_i(I)$. This is specifically true as the cost functions are identical. 
Hence, if we had a polynomial time algorithm for finding an \MMS{} allocation, we would be able to find \MMS{} cost by looking at the maximum cost of an agent.  Consequently, we could decide a 3-partition exists. As a result, finding an \MMS{} allocation is \NPH{}.
\end{proof}

\begin{figure}[t]
    \centering

   \begin{tikzpicture}
        \def\xs{0.8cm} 
        \def\ys{0.5cm} 
        \def\x{0cm} 
        \def\y{0cm} 
        \def\ls{\footnotesize} 
        
        \tikzset{
            node_blank/.style={circle,draw,minimum size=0.5cm,inner sep=0, color=white}, 
            node/.style={circle,draw,minimum size=0.3cm,inner sep=0, fill = black!05},
            node_h/.style={circle,draw,minimum size=0.45cm,inner sep=0, fill = blue!20, font=\footnotesize},
            node_emph/.style={circle, minimum size=0.55cm, black!15, fill = black!15, draw,inner sep=0, font=\footnotesize},
            edge/.style={draw,thick,sloped,-,above,font=\footnotesize},
            arrow/.style={draw, single arrow, minimum width = 0.9cm, minimum height=\y-6*\x+\s, fill=black!10},
            blank/.style={},
        }
        
        \node[node_emph] (_) at (\x + 0*\xs, 0*\ys + \y) {};
        \node[node_h] (h) at (\x + 0*\xs, 0*\ys + \y) {\ls $h$};
        \node[node] (a1) at (\x + -1*\xs, 1*\ys + \y) {};
        \node[node] (a2) at (\x + -2*\xs, 1*\ys + \y) {};
        \node[node] (a3) at (\x + -3*\xs, 1*\ys + \y) {};
        \node[node] (b1) at (\x + -1*\xs, 0*\ys + \y) {};
        \node[node] (b2) at (\x + -2*\xs, 0*\ys + \y) {};
        \node[node] (b3) at (\x + -3*\xs, 0*\ys + \y) {};
        \node[node] (c1) at (\x + -1*\xs, -1*\ys + \y) {};
        \node[node] (c2) at (\x + -2*\xs, -1*\ys + \y) {};
        \node[node] (c3) at (\x + -3*\xs, -1*\ys + \y) {};
        \node[node] (d1) at (\x + 1*\xs, -1*\ys + \y) {};
        \node[node] (d2) at (\x + 2*\xs, -1*\ys + \y) {};
        \node[node] (d3) at (\x + 3*\xs, -1*\ys + \y) {};
        \node[node] (d4) at (\x + 4*\xs, -1*\ys + \y) {};
        \node[node] (d5) at (\x + 5*\xs, -1*\ys + \y) {};
        \node[node] (d6) at (\x + 6*\xs, -1*\ys + \y) {};
        \node[node] (e1) at (\x + 1*\xs, 0*\ys + \y) {};
        \node[node] (e2) at (\x + 2*\xs, 0*\ys + \y) {};
        \node[node] (e3) at (\x + 3*\xs, 0*\ys + \y) {};
        \node[node] (e4) at (\x + 4*\xs, 0*\ys + \y) {};
        \node[node] (e5) at (\x + 5*\xs, 0*\ys + \y) {};
        \node[node] (e6) at (\x + 6*\xs, 0*\ys + \y) {};
        \node[node] (f1) at (\x + 1*\xs, 1*\ys + \y) {};

        \path[edge]
        (h) edge (a1)
        (a1) edge (a2)
        (a2) edge (a3)
        (h) edge (b1)
        (b1) edge (b2)
        (b2) edge (b3)
        (h) edge (c1)
        (c1) edge (c2)
        (c2) edge (c3)
        (h) edge (d1)
        (d1) edge (d2)
        (d2) edge (d3)
        (d3) edge (d4)
        (d4) edge (d5)
        (d5) edge (d6)
        (h) edge (e1)
        (e1) edge (e2)
        (e2) edge (e3)
        (e3) edge (e4)
        (e4) edge (e5)
        (e5) edge (e6)
        (h) edge (f1)
        ;
      
    \end{tikzpicture}
    \caption{{\small An example of a spider (or starlike) graph.
    }}
    \label{fig:mms:hardness}
\end{figure}


\vspace{-3mm}
\section{Characterizing Fair and Efficient Solutions}\label{sec:charac}

The possible incompatibility of fairness and efficiency in our setting was established in \cref{ex:allocations}. Recall that \cref{ex:allocations} gives an instance which does not admit an \EFone{} and \PO{} allocation.  
In this section, we exploit the structure provided by trees to characterize the space of delivery instances for which fair and efficient allocations exist. We first discuss social optimality and then turn our attention to Pareto optimality.


\subsection{Pareto Optimality}
\label{subsec:po}
We first focus on Pareto optimality.
We begin by noting that an $\MMS{}$ and \PO{} allocation always exists, but the same is not true for an $\EFone{}$ and \PO{} allocation.

\begin{restatable}{proposition}{propnonexist}\label{prop:non-existenceEFonePO}
 Given a delivery instance $I = \langle [n], G , h \rangle$, 
 \begin{itemize}
     \item[a.]
     an \EFone{} and \PO{} allocation need not exist,
     \item[b.]
     an \MMS{} and \PO{} allocation always exists.
 \end{itemize}
\end{restatable}

\begin{proof}
    For \textit{a},
    the proof follows from the instance given in \cref{ex:allocations}. As we have discussed earlier, in this example, no allocation is simultaneously \EFone{} and \PO{}. 

    For \textit{b},
    note that as the cost functions are identical across the agents, a leximin optimal allocation will always be \MMS{} and \PO{}. 
    As a result, an \MMS{} and \PO{} allocation always exists because a leximin optimal allocation always exists. 
    
\end{proof}

The central result of this section
is the proof that every \EFone{} and \PO{} allocation
will satisfy \MMS{} as well.
This comes in contrast to typical fair chore division settings, where under additive preferences
\EFone{} and \MMS{} are independent notions
even in the presence of efficiency requirements.
To this end,
we first prove an insightful necessary condition for \EFone{} and \PO{} allocations:
in every such allocation, the pairwise difference in the costs of agents cannot be greater than 1.

\begin{lemma}
    \label{lem:ef1+po}
Given a delivery instance $I = \langle [n], G, h \rangle$ and an \EFone{} allocation $A$,
if $|c(A_i)-c(A_j)|>1$ for some agents $i,j \in [n]$, then $A$ is not \PO{}.
\end{lemma}

\begin{proof}
    Given instance $I$, let $A$ be an \EFone{} allocation.     
    Without loss of generality, 
    assume that $A$ is sorted in non-increasing cost order,
    i.e.,
    $c(A_1) \ge \dots \ge c(A_n)$
    (otherwise we can relabel the agents).
    We will show that if $c(A_n) < c(A_1) - 1$,
    then $A$ is not \PO{}, i.e.,
    it is Pareto dominated by some allocation $A'$
    (not necessarily \EFone{}).

    For every vertex $x \in V \setminus \{h\}$,
    by $p(x)$ let us denote the \emph{parent} of $x$
    in a tree $G$ rooted in $h$.
    Also, for every agent $i \in [n]$,
    let $w(i)$ be the \emph{worst} vertex in $i$'s bundle, i.e.,
    the vertex which on removal gives the largest decrease in cost
    (if there is more than one we take an arbitrary one).
    Formally,
    $$w(i) = \argmax_{x \in A_i} c(A_i) - c(A_i \setminus \{x\}).$$
    
    As $A$ is an \EFone{} allocation, for every agent $i$ with maximal cost, i.e.,
    such that $c(A_i) = c(A_1)$, we have that
    \begin{equation}
        \label{eq:lem:1}
        c(A_i \setminus \{w(i)\}) \le c(A_n) < c(A_i) - 1.
    \end{equation}
    Observe that this is only possible if the parent of $w(i)$ is not serviced by $i$, i.e., $p(w(i)) \not \in A_i$.
    To construct allocation $A'$ which Pareto dominates $A$,
    we look at the agent servicing the parent of the worst vertex of agent 1, call this agent $i_1$ .
    If $i_1$ incurs maximum cost, we look at the agent servicing the parent of the worst vertex of $i_1$.
    We continue in this manner and obtain a maximal sequence of pairwise disjoint agents $1 = i_0, i_1, \dots, i_k$
    such that $p(w(i_{s-1})) \in A_{i_s}$ and $c(A_{i_s}) = c(A_1)$ for every $s \in [k]$.
    
    The cost incurred by the agent servicing the parent of the worst vertex of $i_k$,
    which we denote by $i^*$ (i.e., $p(w(i_k))\in A_{i^*}$) can create two cases.
    Either $i^*$ does not incur maximum cost, i.e., $c(A_{i^*}) < c(A_1)$
    (Case 1), or
    it already appears in the sequence, i.e.,
    $i^* = i_j$ for some $j < k$ (Case 2).\\

    \noindent \underline{\textbf{Case 1}}.
    Consider the allocation $A'$ which is obtained from $A$
    by exchanging the bundles of agent $i_k$
    and agent $i^*$
    with the exception of $w(i_k)$
    (which continues to be serviced by $i_k$).
    See \cref{fig:proofEF1PO} for an illustration.
    Formally, $A'_{i^*} = A_{i_k} \setminus \{w(i_k)\}$,
    $A'_{i_k} \! =\! A_{i^*} \cup \{w(i_k)\}$, and
    $A'_t \! =\! A_t$, for every $t \in [n] \setminus \{i^*\!,i_k\}$.
    Since costs of agents in $[n] \setminus \{i^*,i_k\}$
    are not affected, 
    it suffices to prove that the cost of either $i_k$ or $i^*$
    decreases
    without increasing the other's cost.
    To this end, observe that since parent of $w(i_k)$
    belongs to $A_{i^*}$,
    adding this vertex to $A_{i^*}$ increases the cost by 1, i.e.,
    \begin{equation}
        \label{eq:lem:2}
        c(A'_{i_k}) = c(A_{i^*}) + 1.
    \end{equation}
    Now, let us consider two subcases based on
    the original difference in costs of agents $i_k$ and $i^*$.\\

    \noindent \underline{\textit{Case 1a}}.
    If this difference is greater than one, i.e.,
    $c(A_{i_k}) > c(A_{i^*}) + 1$,
    then from \cref{eq:lem:2} we get that
    $$c(A'_{i_k}) = c(A_{i^*}) + 1 < c(A_{i_k}).$$
    Hence, the cost of $i_k$ decreases.
    
    For $i^*$, from \cref{eq:lem:1}, we have that
    $$c(A'_{i^*}) = c(A_{i_k} \setminus \{w(i_k)\}) \le c(A_n) \le c(A_{i^*}),$$
    so agent $i^*$ does not suffer from the exchange. Consequently, when $c(A_{i_k}) > c(A_{i^*}) + 1$, $A'$ Pareto dominates $A$.\\

    \noindent \underline{\textit{Case 1b}}.
    Otherwise, the difference in costs of $i_k$ and $i^*$ is exactly one,
    i.e., $c(A_{i_k}) = c(A_{i^*}) + 1$.
    Recall from \cref{eq:lem:2} that $c(A'_{i_k})=c(A_{i^*}+1)=c(A_{i_k})$. Thus,
    the cost of $i_k$ stays the same in $A'$.
    
    Now, in order to show that $A'$ Pareto dominates $A$, we need that $c(A_{i^*}')<c(A_{i^*})$. 
    As, $c(A_n) < c(A_1) - 1$, and $c(A_1)=c(A_{i_k})=c(A_{i^*})+1$
    it must be that $c(A_n) < c(A_{i^*})$. 
    Thus, from \cref{eq:lem:1} we have
    $$c(A'_{i^*}) \le c(A_n) < c(A_{i^*}),$$ 
    i.e., the cost of $i^*$ decreases under $A'$.
    As a result, even when $c(A_{i_k}) = c(A_{i^*}) + 1$, $A'$ Pareto dominates $A$.\\
    
\begin{figure}[t]
    \centering
\begin{tikzpicture}
        \def\xs{0.9cm} 
        \def\ys{0.9cm} 
        \def\x{0cm} 
        \def\y{0cm} 
        \def\ls{\footnotesize} 
        
        \tikzset{
            node/.style={circle,draw,minimum size=0.3cm,inner sep=0, fill = black!05},
            edge/.style={draw,thick,sloped,-,above,font=\footnotesize},
            arrow/.style={draw, single arrow, minimum width = 0.5cm, minimum height=\y-6*\x+\s, fill=black!10},
            operation/.style={sloped,>=stealth,above,font=\footnotesize},
            blank/.style={},
            alloc_a/.style={circle, draw, color=blue!30, minimum size=0.5cm, line width=2.5},
            alloc_b/.style={rectangle, draw, color=red!30, minimum size=0.45cm, line width=2.5},
            alloc_c/.style={regular polygon, regular polygon sides=5, draw, color=green!50, minimum size=0.55cm, line width=2.5},
        }

        \node[blank] (0) at (\x, -0.3*\ys + \y) {};
        \node[node] (1) at (\x, -1*\ys + \y) {};
        \node[node] (2) at (\x, -2*\ys + \y) {};
        \node[node] (2b) at (\x + -0.9*\xs, -2*\ys + \y) {};
        \node[node, label = {[label distance=0.07cm] 270: \ls $w(i_k)$}] (3) at (\x, -3*\ys + \y) {};
        \node[blank] (label) at (\x + - 1*\xs, -0.7*\ys + \y) {$A$};
        \node[blank] (arrow_start) at (\x + 0.4*\xs, -2*\ys + \y) {};
        \node[blank] (subcaption) at (\x + 1.1cm, -4.3*\ys + \y) {\underline{\textbf{Case 1}}};

        \path[edge]
        (0) edge (1)
        (1) edge (2)
        (1) edge (2b)
        (2) edge (3)
        ;
        
        \node[alloc_a] (A1) at (\x, -1*\ys + \y) {};
        \node[alloc_a] (A2) at (\x, -2*\ys + \y) {};
        \node[alloc_b] (A2b) at (\x + -0.9*\xs, -2*\ys + \y) {};
        \node[alloc_b] (A3) at (\x, -3*\ys + \y) {};

        \def\x{2.5cm} 
        \node[blank] (0) at (\x, -0.3*\ys + \y) {};
        \node[node] (1) at (\x, -1*\ys + \y) {};
        \node[node] (2) at (\x, -2*\ys + \y) {};
        \node[node] (2b) at (\x + -0.9*\xs, -2*\ys + \y) {};
        \node[node, label = {[label distance=0.07cm] 270: \ls $w(i_k)$}] (3) at (\x, -3*\ys + \y) {};
        \node[blank] (label) at (\x + - 1*\xs, -0.7*\ys + \y) {$A'$};
        \node[blank] (arrow_end) at (\x - 1.4*\xs, -2*\ys + \y) {};

        \path[edge]
        (0) edge (1)
        (1) edge (2)
        (1) edge (2b)
        (2) edge (3)
        ;
        
        \node[alloc_b] (A1) at (\x, -1*\ys + \y) {};
        \node[alloc_b] (A2) at (\x, -2*\ys + \y) {};
        \node[alloc_a] (A2b) at (\x + -0.9*\xs, -2*\ys + \y) {};
        \node[alloc_b] (A3) at (\x, -3*\ys + \y) {};


        \path[->,draw,very thick]
        (arrow_start) edge[operation] (arrow_end)
        ;

        \def\x{6cm}
        \node[blank] (0) at (\x, -0.3*\ys + \y) {};
        \node[node] (1) at (\x, -1*\ys + \y) {};
        \node[node] (2) at (\x, -2*\ys + \y) {};
        \node[node] (2b) at (\x + -1.1*\xs, -2*\ys + \y) {};
        \node[node] (2c) at (\x + 1.1*\xs, -2*\ys + \y) {};
        \node[node, label = {[label distance=0.07cm] 270: \ls $w(i_2)$}] (3) at (\x, -3*\ys + \y) {};
        \node[node, label = {[label distance=0.07cm] 270: \ls $w(i_1)$}] (3b) at (\x + -1.1*\xs, -3*\ys + \y) {};
        \node[node, label = {[label distance=0.07cm] 270: \ls $w(i_3)$}] (3c) at (\x + 1.1*\xs, -3*\ys + \y) {};
        \node[blank] (label) at (\x + - 1*\xs, -0.7*\ys + \y) {$A$};
        \node[blank] (arrow_start) at (\x + 1.4*\xs, -2*\ys + \y) {};
        \node[blank] (subcaption) at (\x + 1.75cm, -4.3*\ys + 1*\y) {\underline{\textbf{Case 2}}};

        \path[edge]
        (0) edge (1)
        (1) edge (2)
        (1) edge (2b)
        (1) edge (2c)
        (2) edge (3)
        (2b) edge (3b)
        (2c) edge (3c)
        ;

        \node[alloc_a] (A1) at (\x, -1*\ys + \y) {};
        \node[alloc_a] (A2) at (\x, -2*\ys + \y) {};
        \node[alloc_b] (A2b) at (\x + -1.1*\xs, -2*\ys + \y) {};
        \node[alloc_c] (A2c) at (\x + 1.1*\xs, -2*\ys + \y) {};
        \node[alloc_b] (A3) at (\x, -3*\ys + \y) {};
        \node[alloc_c] (A3b) at (\x + -1.1*\xs, -3*\ys + \y) {};
        \node[alloc_a] (A3c) at (\x + 1.1*\xs, -3*\ys + \y) {};

        \def\x{9.5cm} 
        \node[blank] (0) at (\x, -0.3*\ys + \y) {};
        \node[node] (1) at (\x, -1*\ys + \y) {};
        \node[node] (2) at (\x, -2*\ys + \y) {};
        \node[node] (2b) at (\x + -1.1*\xs, -2*\ys + \y) {};
        \node[node] (2c) at (\x + 1.1*\xs, -2*\ys + \y) {};
        \node[node, label = {[label distance=0.07cm] 270: \ls $w(i_2)$}] (3) at (\x, -3*\ys + \y) {};
        \node[node, label = {[label distance=0.07cm] 270: \ls $w(i_1)$}] (3b) at (\x + -1.1*\xs, -3*\ys + \y) {};
        \node[node, label = {[label distance=0.07cm] 270: \ls $w(i_3)$}] (3c) at (\x + 1.1*\xs, -3*\ys + \y) {};
        \node[blank] (label) at (\x + - 1*\xs, -0.7*\ys + \y) {$A'$};
        \node[blank] (arrow_end) at (\x - 1.4*\xs, -2*\ys + \y) {};

        \path[edge]
        (0) edge (1)
        (1) edge (2)
        (1) edge (2b)
        (1) edge (2c)
        (2) edge (3)
        (2b) edge (3b)
        (2c) edge (3c)
        ;
        
        \node[alloc_a] (A1) at (\x, -1*\ys + \y) {};
        \node[alloc_a] (A2) at (\x, -2*\ys + \y) {};
        \node[alloc_b] (A2b) at (\x + -1.1*\xs, -2*\ys + \y) {};
        \node[alloc_c] (A2c) at (\x + 1.1*\xs, -2*\ys + \y) {};
        \node[alloc_a] (A3) at (\x, -3*\ys + \y) {};
        \node[alloc_b] (A3b) at (\x + -1.1*\xs, -3*\ys + \y) {};
        \node[alloc_c] (A3c) at (\x + 1.1*\xs, -3*\ys + \y) {};

        \path[->,draw,very thick]
        (arrow_start) edge[operation] (arrow_end)
        ;

    \end{tikzpicture}
    \caption{\small{ Illustrating the proof of Prop.  \ref{lem:ef1+po}.
    In
    Case 1,  $i_k$
    (red squares)
    exchanges its bundle with $i^*$
    (blue circles)
    except for
    $w(i_k)$.
    In
    Case 2,
     $i_1$, $i_2$, and $i_3$
    (pentagons, squares, and circles, resp.)
    swap their worst vertices along the cycle.}}
    \label{fig:proofEF1PO}
\end{figure}

    \noindent \underline{\textbf{Case 2}}.
    When $i^*=i_j$ for some $j < k$,
    we have a cycle of agents
    $i^* = i_j, i_{j+1}, \dots, i_k$ such that $c(A_{i_s}) = c(A_1)$
    and $p(w(i_{s})) \in A_{i_{s+1}}$ for every $s \in \{j,\dots,k\}$
    (we denote $i_j$ as $i_{k+1}$ as well for notational convenience).
    Here, we consider two subcases,
    based on whether it happens somewhere in the cycle
    that the parent of the worst vertex of one agent
    is the worst vertex of the next agent,
    i.e., $p(w(i_{s})) = w(i_{s+1})$ for some $s \in \{j,\dots,k\}$.\\

    \noindent \underline{\textit{Case 2a}}.
    We first look at the subcase where there does exist $s \in \{j,\cdots, k\}$ s.t. $p(w(i_s))=w(i_{s+1})$.
    Here, we consider allocation $A'$ which is obtained from $A$ by giving $w(i_{s+1})$ to agent $i_s$.
    Formally, $A'_{i_{s+1}} = A_{i_{s+1}} \setminus \{w(i_{s+1})\}$,
    $A'_{i_s} = A_{i_s} \cup \{w(i_{s+1})\}$, and
    $A'_t = A_t$, for every $t \in [n] \setminus \{i_s,i_{s+1}\}$.
    
    Now, the cost of agent $i_{s+1}$ decreases as it no longer services its worst vertex.
    Observe that agent $i_s$ was already servicing a child of $w(i_{s+1})$ in $A$.
    Consequently, it was visiting $w(i_{s+1})$ on the way.
    Hence, the cost of $i_s$ stays the same.
    As the bundles of the remaining agents did not change,
    $A'$ Pareto dominates $A$. \\

    \noindent \underline{\textit{Case 2b}}.
    Finally, we have the case where there is no agent in the cycle for which the parent of its worst vertex is the worst vertex of the next agent. 
    In this case we consider the allocation $A'$ obtained from $A$ by swapping the worst vertices of the agents along the cycle
    (illustrated in \cref{fig:proofEF1PO}).
    Formally, $A'_{i_s} = (A_{i_s} \setminus \{w(i_s)\}) \cup \{w(i_{s-1})\}$
    for every $s \in \{j+1,\dots,k+1\}$ and $A'_i = A_i$ for every $i \in [n] \setminus \{i_{j+1},\dots,i_{k+1}\}$.
    
    As each  $i_s$ is servicing the parent of the worst vertex of the previous agent in the cycle, i.e., $p(w(i_{s-1})) \in A'_{i_s}$,
    servicing  $w(i_{s-1})$ incurs an additional cost of 1.
    However, from \cref{eq:lem:1} giving away the worst vertex decreases the cost by more than 1.
    Hence, the cost of each agent in the cycle decreases.
    The other agents' costs stay the same, 
    thus $A'$ Pareto dominates $A$.
\end{proof} 

We can now show that \EFone{} and \PO{} imply \MMS{}.

\begin{restatable}{theorem}{efonepoimpliesmms} \label{thrm:ef1+po:implies_mms}
Given a delivery instance $I= \langle [n], G, h \rangle$, an \EFone{} and \PO{} allocation exists if and only if for any leximin optimal allocation $A$, $\max_{i,j\in [n]}|c(A_i)-c(A_j)|\leq 1$. Further, every \EFone{} and \PO{} allocation satisfies \MMS{}. 
\end{restatable}

\begin{proof}
    We begin by proving  that an \EFone{} and \PO{} allocation $A$ exists if and only if for any leximin optimal allocation $A$, we have that $\max_{i,j\in [n]}|c(A_i)-c(A_j)|\leq 1$. Observe that by \cref{lem:ef1+po}, we have that under a \EFone{} and \PO{} allocation $A$, it must be that $\max_{i,j\in [n]}|c(A_i)-c(A_j)|\leq 1$. We shall now that $A$ is necessarily leximin optimal. 
    
    Without loss of generality,
    let $A$ be s.t.
    $c(A_1) \ge \dots \ge c(A_n)$.
    For contradiction, assume that there exists $A'$
    that leximin dominates $A$. As the agents have identical costs and we are only comparing for leximin domination, we can assume without loss of generality that $c(A_1')\geq \dots \geq c(A_n')$.
    Thus, there exists $i \in [n]$ such that
    $c(A'_i) < c(A_i)$ and $c(A'_j)= c(A_j)$ for every $j \in [i-1].$.
    
    Fix an agent $j \in [n]$ such that $j >i$ and thus $c(A_j)\leq c(A_i)$.
    From \cref{lem:ef1+po} we know that it must be that
    either $c(A_j) = c(A_i)$ or $c(A_j) = c(A_i)-1$.
    Recall that $A'$ is also sorted and thus $c(A'_j) \le c(A'_i) < c(A_i)$.
    As a result,  it must be that $c(A'_j) < c(A_i)-1 \le c(A_j)$.
    Thus, $A'$ also Pareto dominates $A$,
    which contradicts the assumption that $A$ is \PO{}.

    Now for the converse, let $A$ be a leximin optimal allocation such that $\max_{i,j\in [n]}|c(A_i)-c(A_j)|\leq 1$. Clearly, $A$ is \PO{}. We shall now show that $A$ is \EFone{}.  Fix $i,j\in [n]$, arbitrarily. We shall show that $i$ is \EFone{} towards $j$. Firstly, if $c(A_i)\leq c(A_j)$, clearly $i$ does not envy $j$. From \cref{lem:ef1+po}, the only remaining case is that $c(A_i)=c(A_j)+1$. Choose a leaf $x$ in the graph $G|A_i$. That is $x\in A_i$ is such that $c(A_i\setminus \{x\})\leq c(A_i)-1=c(A_j)$. Consequently, $A$ is \EFone{}.

    As a result, an \EFone{} and \PO{} allocation will always be leximin optimal. Hence, whenever an \EFone{} and \PO{} allocation exists, it must also be \MMS{}.
\end{proof}


Finally, by modifying the proof of \cref{thm:mms:hardness}
we show that deciding if there exists a \PO{}
allocation that is \EFone{}
is also \NPH{}. 
Also, we note that
hardness for \MMS{} and \PO{}
is directly implied by
\cref{thm:mms:hardness}. 

\begin{restatable}{proposition}{proppohardness}\label{prop:po:hardness}
Given a delivery instance $I=\langle [n],G,h\rangle$,
checking whether there exists an \EFone{} and \PO{} or
finding an \MMS{} and \PO{} allocation is \NPH{}.
\end{restatable}

\begin{proof}
We use the reduction from \cref{thm:mms:hardness} to prove this result. Firstly, consider an \MMS{} and \PO{} allocation.

\paragraph{\MMS{} and \PO{}.} Observe that the \MMS{} allocations in the case when the original instance does admit a 3-Partition  must also be \PO. Hence, finding an \MMS{} and \PO{} allocation  is \NPH{}.

\paragraph{\EFone{} and \PO{}.} For the same reduction, we now show hardness of \EFone{} and \PO{} allocations.  Recall that from \cref{lem:ef1+po} for any \EFone{} and \PO{} allocation, $A=(A_1,\ldots, A_n)$, for any $i,j$, it must be that $|c(A_i)-c(A_j)|\leq 1$.  
As $\sum_{i=1}^n c(A_i)=nT$, there exist $i$ s.t. $c(A_i)<T$ if and only if there exists $i'$ such that $c(A_{i'})>T$. Thus, $c(A_{i'})-c(A_i)\geq 2$ and for such an $i$, \EFone{} is violated. Consequently, an \EFone{} and \PO{} allocation must give cost $T$ to all agents.

Hence, an \EFone{} and \PO{} allocation exists if any only if a there exists a 3-Partition. Consequently, finding an \EFone{} and \PO{} allocation when one exists is \NPH{}.
\end{proof}

\subsection{Social Optimality}
\label{subsec:so}

We now give more specific characterizations for fair and socially optimal allocations.  Recall that in every socially optimal allocation, each edge is traversed by a unique agent, thus the sum of all agents' costs is $m$ (the number of edges in a tree with $m+1$ vertices).
This is a very demanding condition, hence it is not surprising that it may be impossible to satisfy it together with some fairness requirement.
We show that checking if there exists a social optimal and fair allocation is computationally hard. 

\begin{restatable}{proposition}{propsoexist}\label{prop:soexist}
Given a delivery instance $I=\langle [n], G ,h\rangle$, an \SO{} allocation that satisfies \EFone{} or \MMS{} need not exist. Moreover, checking whether an instance admits such an allocation is \NPH{}.
\end{restatable}

\begin{proof}
\cref{ex:allocations} shows that there can exist instances where no \SO{} allocation satisfies \MMS{} or \EFone{}. Now recall the proof of \cref{thm:mms:hardness} and \cref{prop:po:hardness}. We find that the same reduction helps us prove the intractability of finding fair and \SO{} allocations. We shall show that in the instance constructed in the reduction, a fair and \SO{} allocation exists if and only if there exists a 3-Partition.

We have already shown that finding fair and \PO{} allocations whenever they exist is \NPH{} in \cref{prop:soexist}. Using the same reduction (which is the one \cref{thm:mms:hardness}), we can show that in the instances constructed, we can show that any \PO{} allocation must be \SO{}. This will directly imply that finding a fair and \SO{} allocation is \NPH{}.

Let $A$ be an allocation that is not \SO{}. Thus, there must exist a vertex $v_j^t$ that is visited by multiple agents. Clearly, $v_j^t$ is not a leaf. 

Recall that the reduction constructs a spider graph, where we create a branch of $s_j$ vertices for each $s_j$ in the given \textsc{\textup{3-Partition}} instance. The vertices on branch $j$ are labelled $\{v_j^1,\cdots,v_j^{s_j}\}$ with $v_j^{s_j}$ being the leaf of the branch.

Let $v_j^{s_j}$ be serviced by agent $i$. Consider allocation $A'$ which is identical to $A$ with the exception that $i$ services all of $v_j^1,\cdots, v_j^{s_j}$. Clearly, the cost of $i$ is the same as that in $A$, but the costs of all other agents either decrease or remain the same. As multiple agents, including $i$, visit $v_j^t$, the cost of at least one agent reduces. Consequently, $A'$ Pareto dominates $A$. Thus, every \PO{} allocation in our constructed graph must be \SO{}.

Hence, any polynomial time algorithm to find an \MMS{} and \SO{} allocation or find an \EFone{} and \SO{} allocation when one exists, would also find the corresponding \PO{} allocations and solve 3-Partition. Thus, we have the proposition. 



\end{proof}

Despite the computational hardness result, we exploit the tree structure to characterize instances with such allocations.

\begin{restatable}{proposition}{propso}\label{prop:so}
    Given a delivery instance $I=\langle [n], G,  h \rangle$,
    an allocation $A$ is SO,
    if and only if,
    every branch, $B$, outgoing from $h$,
    is fully contained in a bundle of some agent $i \in [n]$,
    i.e., $B \subseteq A_i$.
\end{restatable}

\begin{proof}
    Observe that the sum of the costs of all agents
    can never be smaller than
    the number of edges in the graph, i.e., $|E| = m$.
    Indeed, since every vertex has to be serviced by some agent,
    each edge must appear in $G|_{A_i}$ for some $i \in [n]$.
    Moreover, if every branch outgoing from the hub
    is fully contained in some bundle,
    every edge appears in $G|_{A_i}$ for exactly one $i \in [n]$.
    Hence, in every allocation satisfying the assumption,
    the total costs is equal to $m$
    and thus the allocation is \SO{}.

    Now, consider an allocation in which there exists a branch, $B$,
    and agents $i \neq j$ such that
    $B \cap A_i \neq \emptyset$ and $B \cap A_j \neq \emptyset$.
    Let $u$ be a vertex in $B$ connected to the hub, $h$.
    Then, edge $(h,u)$ appears in both $G|_{A_i}$ and $G|_{A_j}$.
    Since every other edge appears in $G|_{A_k}$
    for at least one $k \in [n]$,
    we get that the total cost is greater than $m$.
    Thus, such an allocation is not \SO{}.
\end{proof}

\begin{restatable}{theorem}{thrmsocharacterization}\label{thm:so:characterization}
Given a delivery instance $I = \langle [n], G, h \rangle$:
\begin{itemize}
    \item[a.] An \EFone{} and \SO{} allocation exists if and only if there is a partition $(P_1,\dots,P_n)$ of branches out of $h$ such that $\textstyle\sum_{B \in P_i} |B| - \textstyle\sum_{B \in P_j} |B| \le 1$, for every $i,j \in [n]$, 
    \item[b.] An \MMS{} and \SO{} allocation exists if and only if there is a partition $(P_1,\dots,P_n)$ of branches outgoing from $h$ such that $\textstyle\sum_{B \in P_i} |B| \le MMS_i(I)$ for every $i \in [n]$.
\end{itemize}
\end{restatable}

\begin{proof}
    From~\cref{prop:so}, we know that each \SO{} allocation must be a partition of whole branches outgoing from $h$.
    Hence, \textit{b} follows from the definition of \MMS{}.

    For \textit{a},
     observe that if an agent's bundle consists of a union of whole branches, then removing a vertex from its bundle reduces the cost of the agent by 1, if the vertex is a leaf, or by 0, otherwise. Hence, in order to achieve \EFone{} the costs of agents cannot differ by more than one.
\end{proof}

By \cref{prop:soexist},
checking the conditions in  the above theorem is \NPH{}.
However, we develop a polynomial-time verifiable necessary condition for \EFone{} and \SO{} existence
using the notion of the \emph{center} of the graph
(which was studied in computational social choice~\citep{skibski2023closeness} and in theoretical computer science in general~\citep{goldman1971optimal}).

\begin{definition}
The \emph{center} of a tree $G=(V,E)$ is a 
set of vertices
$C = \argmin_{v \in V} \sum_{u \in V} \dist(u,v)$,  
where $\dist(x,y)$ is the length of a shortest path from $x$ to $y$.
\end{definition}

\begin{restatable}{proposition}{propcenterefoneso}\label{prop:center:ef1so}
Given a delivery instance $I = \langle [n], G, h \rangle$, there exists an \EFone{} and \SO{} allocation only if the hub is in the center of the tree.
\end{restatable}

\begin{proof}
Let us show that 
if for some vertex $v$, there exists a partition of the branches outgoing from $v$ such that the differences between the number of vertices in each bundle are equal or less than 1, then $v$ has to be in the center.
In conjunction with \cref{thm:so:characterization}a,
this will imply the thesis.

For contradiction,
assume that there exists a tree $G = (V,E)$
and vertex $v$ such that there is
a partition of the branches
outgoing from $v$ such that
the difference in
the number of vertices 
in each pair of bundles is 
equal or less than 1,
but at the same time there exists $u \in V$
such that
\[  
    \sum_{w \in V} \dist(w,v) > \sum_{w \in V} \dist(w,u).
\]
Let $B$, be the branch outgoing from $v$ that contains vertex $u$.
From the fact that the branches outgoing from $v$
can be partitioned in such a way that
the number of the vertices in each part
is at most equal to one plus the number of vertices in each other part,
we get that in particular
\(
    |V(B)| \le |V \setminus V(B) \setminus \{v\}| + 1.
\)
This implies that
\begin{equation}
    \label{eq:centrality}
    |V(B)| \le |V \setminus V(B)|.
\end{equation}

Let $v = v_0, v_1,\dots,v_k = u$
be the path from vertex $v$ to vertex $u$.
Since sum of distances to $v$ is greater
than sum of distances to $u$,
there has to be $i \in [k]$
such that
\[  
    \sum_{w \in V} \dist(w,v_{i-1}) > \sum_{w \in V} \dist(w,v_{i}).
\]
For arbitrary $w\in V$ consider the difference 
\(
    \dist(w,v_{i-1}) - \dist(w,v_{i}).
\)
Since $v_{i-1}$ and $v_{i}$
are adjacent vertices,
this difference has value $1$,
if $w$ is closer to $v_{i}$ than $v_{i-1}$,
or $-1$, otherwise.
Observe that every vertex in $V \setminus V(B)$
is closer to $v_{i-1}$ than $v_{i}$.
Thus, from \cref{eq:centrality} we get that
\begin{align*}
    \sum_{w \in V} \dist(w,\!v_{i-1}) -  \sum_{w \in V} \dist(w, v_{i}) \le
    |V(B)| - |V  \setminus V(B)|
    \le 0, 
\end{align*} 
which is a contradiction. As a result, there cannot exist a vertex $v$ that is not the center of the graph but there is a partition of branches outgoing from $v$ such that the absolute difference in  the number of vertices in each pair of bundles is at most one.  
\end{proof}

\subsection{EFX allocations}\label{sec:efx}
We now consider a fairness notion that is stronger than \EFone{}, but is still a relaxation of envy-freeness
is \emph{envy-freeness up to any item} (\EFX). \EFX{} has proven to be one of the most enigmatic fairness notions with the existence of \EFX{} allocations being an open problem in most settings.

In our setting,
allocation $A$ is \EFX{} if
for every pair of agents $i,j \in [n]$
and every vertex $x \in A_i$,
it holds that $c(A_i \setminus \{x\}) \leq c(A_j)$.
For chore division settings, \citet{barman2023fair} have shown that when allocating chores with identical monotone costs, an \EFX{}
allocation always exists and
can be computed in pseudo-polynomial time.
Observe that our setting is one where each vertex is a chore, and costs are identical submodular functions. Consequently, an \EFX{} allocation must always exist. Further, in our case the costs of agents are polynomial
with respect to the size of the input.
Thus,
the algorithm given by \citet{barman2023fair} runs in polynomial time.

While \EFX{} allocations always exist in our setting, efficient \EFX{} allocations need not. This can be seen clearly as an \EFX{} allocation must also satisfy \EFone{} and efficient \EFone{} allocations need not exist. 
Although \EFX{} is more restrictive than \EFone{}, 
it seems less compatible with efficiency in delivery settings.
To see this, observe that 
in an \EFX{} allocation
any envious agent, $i$, cannot posses a vertex, $x$,
that is on a path to another of its vertices.
Otherwise, removing $x$ from the bundle of $i$
would not decrease $i$'s cost, which would contradict \EFX{}.
In fact, we show that for large classes of graphs,
\EFX{} combined with some efficiency requirement
is as restrictive as envy-freeness.

\begin{restatable}{theorem}{thrmefx}\label{thrm:efx}
    Given an instance $I = \langle [n], G, h \rangle$
    such that $m \ge n$ and an \EFX{} allocation $A$, if
    \begin{itemize}
        \item[a.] the hub is a parent of at most one leaf
        and $A$ is \SO{} or,
        \item[b.] every vertex is a parent of at most one leaf
        and $A$ is \PO{},
    \end{itemize}
then $A$ satisfies \EF.
\end{restatable}
 
\begin{proof}
    Consider an arbitrary instance
    $I = \langle [n], G, h \rangle$
    and an \EFX{} allocation $A$.
    We shall establish both conditions
    \textit{a} and \textit{b}, separately.
    
    \textbf{Part a.} For {\em a}, let us assume that $A$ is both \EFX{} and \SO{} and $I$ is such that $h$ is the parent of at most one leaf.
    Now, assume for contradiction
    that there exists an agent,
    $i \in [n]$, that envies another agent.
    Observe that since $m \ge n$,
    we must have that $c(A_i) > 1$.
    Since the hub is the parent
    of at most one leaf,
    there must be a vertex in $x \in A_i$
    that is not a leaf adjacent to $h$.
    As $A$ is \SO{},
    then from \cref{prop:so}
    the whole branch containing $x$ is given to $i$.
    Let $y$ be a vertex of this branch adjacent to $h$.
    The branch has more then one vertex,
    so removing $y$ does not decrease the cost of agent $i$,
    i.e., $c(A_i \setminus \{y\}) = c(A_i)$.
    But that contradicts the \EFX{}.
    
    \textbf{Part b.} For \textit{b}, we assume that $I$ is such that each vertex in $G$ is the parent of at most one leaf and $A$ is an \EFX{} and \PO{} allocation. We shall first show that all agents that envy others must only service leaves.
    Assume otherwise, i.e., $i$ is envious
    and $x \in A_i$ is not a leaf.
    If there exists a vertex, $y$, that is a descendant
    of $x$ in rooted tree $(G,h)$,
    then removing $x$ would not decrease
    the cost of agent $i$
    (as $x$ is visited by $i$ either way).
    But this violates \EFX{}.
    On the other hand,
    if $A_i$ does not contain any descendants of $x$,
    then let us denote arbitrary child of $x$ by $z$
    and by $j$ an agent that serves $z$, i.e., $z \in A_j$.
    Observe that removing vertex $x$
    from the bundle of $i$ and giving it to $j$
    would decrease the cost of $i$,
    but not increase the cost of $j$
    (as it visits $x$ either way)
    nor any other agent (as their bundles do not change).
    Hence, this constitutes
    a Pareto improvement,
    which violates \PO{}.

    Now, assume for contradiction
    that there exists an agent, $i \in [n]$,
    that envies another agent.
    From the previous paragraph,
    we know that $A_i$ consists of leaves only.
    Let us take a leaf $x \in A_i$
    that is the furthest from the hub $h$.
    Let $y$ be the parent of $x$ and
    $j$ an agent that serves it, i.e., $y \in A_j$.
    From the previous paragraph we know that
    $j$ is not envious of any other agent.
    Thus, $c(A_j) < c(A_i)$.
    As \EFX{} implies \EFone{},
    by \cref{prop:ef1}, this means that
    $c(A_j) = c(A_i) - 1$.
    Now, consider allocation $A'$ obtained from
    $A$ by giving $j$ all of $i$'s vertices except for $x$
    and to $i$ all vertices of $j$ and $x$.
    Formally, $A'_i = A_j \cup \{x\}$,
    $A'_j = A_i \setminus \{x\}$, and
    $A'_k = A_k$, for every $k \in [n] \setminus \{i,j\}$.
    Observe that adding vertex $x$ to bundle $A_j$
    increases its cost by 1.
    Thus,
    \(
        c(A'_i) = c(A_j) + 1 = c(A_i),
    \)
    which means that the cost for agent $i$ did not increase.
    On the other hand,
    since $A_i$ contains only leaves,
    $x$ is the leaf in $A_i$ that is the furthest from the hub,
    and $y$ is not a parent of any other leaf,
    an agent serving $A_i \setminus \{x\}$
    does not have to visit $y$.
    Hence, removing vertex $x$ from $A_i$ decreases the cost
    of at least 2, i.e.,
    \(
        c(A'_j) \le c(A_i) - 2 = c(A_j) -1.
    \)
    Hence, the cost of agent $j$ decreased.
    Since the cost of the remaining agents is unchanged,
    $A'$ Pareto dominates $A$---a contradiction.
\end{proof}
\begin{figure}[t]
    \centering

   \begin{tikzpicture}
        \def\xs{1.2cm} 
        \def\ys{0.4cm} 
        \def\x{0cm} 
        \def\y{0cm} 
        \def\ls{\footnotesize} 
        
        \tikzset{
            node_blank/.style={circle,draw,minimum size=0.5cm,inner sep=0, color=white}, 
            node/.style={circle,draw,minimum size=0.45cm,inner sep=0, fill = black!05},
            node_h/.style={circle,draw,minimum size=0.6cm,inner sep=0, fill = blue!20, font=\footnotesize},
            node_emph/.style={circle, minimum size=0.75cm, black!15, fill = black!15, draw,inner sep=0, font=\footnotesize},
            edge/.style={draw,thick,sloped,-,above,font=\footnotesize},
            arrow/.style={draw, single arrow, minimum width = 0.9cm, minimum height=\y-6*\x+\s, fill=black!10},
            blank/.style={},
            alloc_a/.style={circle, draw, color=blue!30, minimum size=0.8cm, line width=2.5},
            alloc_b/.style={rectangle, draw, color=red!30, minimum size=0.7cm, line width=2.5},
        }
        
        \node[node_emph] (_) at (\x + 1*\xs, 0*\ys + \y) {};
        \node[node] (a) at (\x + -1*\xs - 0.1cm, 0*\ys + \y) {\ls $a$};
        \node[node] (b) at (\x + -0*\xs - 0.1cm, 0*\ys + \y) {\ls $b$};
        \node[node_h] (h) at (\x + 1*\xs, 0*\ys + \y) {\ls $h$};
        \node[node] (c) at (\x + 2*\xs + 0.1cm, 0*\ys + \y) {\ls $c$};
        \node[node] (d) at (\x + 3*\xs + 0.1cm, 0*\ys + \y) {\ls $d$};
        \node[node] (e) at (\x + 4*\xs + 0.1cm, 0*\ys + \y) {\ls $e$};
        
        \node[alloc_a] (_) at (\x + -1*\xs - 0.1cm, 0*\ys + \y) {};
        \node[alloc_a] (_) at (\x + -0*\xs - 0.1cm, 0*\ys + \y) {};
        \node[alloc_b] (_) at (\x + 2*\xs + 0.1cm, 0*\ys + \y) {};
        \node[alloc_b] (_) at (\x + 3*\xs + 0.1cm, 0*\ys + \y) {};
        \node[alloc_b] (_) at (\x + 4*\xs + 0.1cm, 0*\ys + \y) {};
        
        \path[edge]
        (a) edge (b)
        (b) edge (h)
        (h) edge (c)
        (c) edge (d)
        (d) edge (e)
        ;

    \end{tikzpicture}
    \caption{An example graph with an \EFone{} and \PO{} allocation marked. There is no \EFX{} and \PO{} allocation for this graph and two agents.}
    \label{fig:no:efx+po}
\end{figure}

While \EFone{} is also not always compatible with \PO{}
it is much less restrictive then \EFX{}.
In particular, \cref{fig:no:efx+po} shows an example of an instance in which there is no \EFX{} and \PO{} allocation
but \EFone{} and \PO{} allocation does exist.
Indeed, for this instance
there exist two different \PO{} allocations
(up to permutation of agents), i.e.,
the allocation shown in \cref{fig:no:efx+po}
and the allocation in which
one agent serves all vertices.
Neither of them is \EFX{}. 

Finally to conclude, the discussion of \EFX{}, we observe that as any \EFX{} allocation will be \EFone{}. In particular, in the reduction in \cref{prop:po:hardness}, an \EFX{} and \PO{} allocation will exist if and only if an \EFone{} and \PO{} allocation exists. Consequently, we get the following result

\begin{restatable}{proposition}{propefxpohardness}\label{prop:efx+po:hardness}
Given a delivery instance $I=\langle [n],G,h\rangle$, it is \NPH{} to 
check whether there exists i) an \EFX{} and \SO{}  allocation and ii) an \EFX{} and \PO{} allocation.
\end{restatable}

\begin{proof}
    We develop on the proof of \cref{prop:po:hardness} and \cref{prop:so}. As any \EFX{} allocation must satisfy \EFone{}, analogous to the proof in \cref{prop:po:hardness} we have from \cref{lem:ef1+po} that  an \EFX{} and \PO{} allocation to exists if and only if the given instance admits a 3-Partition. Thus, finding an \EFX{} and \PO{} allocation whenever it exists is \NPH{}.

    As shown in the proof of \cref{prop:so}, in the constructed instance an allocation is \PO{} if and only if it is socially optimal. As a result, finding an \EFX{} and \SO{} allocation, whenever it exists, is \NPH{}.
\end{proof}



\section{Computing Fair and Efficient Solutions}\label{sec:xpalg}

We have shown that finding a fair and efficient allocation (or deciding if it exists) is computationally intractable. 
In this section, we develop a recursive algorithm for computing each combination of the fairness and efficiency notions we consider. This algorithm is XP with respect to the number of agents. That is, when the number of agents is bounded by a constant, the running time of our algorithm is polynomial. 

Given a delivery instance, our algorithm (\cref{alg:ef1+po:main}) finds a set of \PO{} allocations such that all possible cost distributions under a \PO{} allocation are covered. We call this set, the Pareto frontier of the instance.  

\begin{definition}
    Given a delivery instance $I = \langle [n], G, h \rangle$, its \emph{Pareto frontier} is a minimal set of allocations $\mathcal{F}$ such that for every \PO{} allocation $A$ there exists $B \in \mathcal{F}$ and permutation of agents $\pi$ such that $c(A_i) = c(B_{\pi(i)})$, for every $i \in [n]$.
\end{definition}

Clearly, a leximin optimal solution will be part of the Pareto frontier. Thus, we can use this algorithm to find an \MMS{} and \PO{} allocation. Further, we can use the tools built in \cref{sec:charac} to identify if an allocation satisfying any other combination of fairness and efficiency exists.

\paragraph{Algorithm Overview.}
Throughout the algorithm, we keep allocations in the list $\mathcal{F}$, with each allocation sorted in non-increasing cost order. That is, for each $A\in \mathcal{F}$, $c(A_1)\geq c(A_2)\geq \cdots c(A_n)$. The algorithm finds the Pareto frontier of a given tree rooted at $h$, by recursively finding the respective Pareto frontiers of the subtrees rooted at each of the children of $h$ and combining them to generate \PO{} allocations for the tree rooted at $h$. Specifically, this proceeds as follows:

First, the set $\mathcal{F}$ is initialized with just one empty allocation.
Then, we look at all the vertices directly connected to the hub.
For each, say $u$, we run our algorithm on a smaller instance
where the graph is just the branch outgoing from $h$
that $u$ is on, and $u$ is the hub. Let $\mathcal{F}_u$ be the Pareto frontier of this new tree, found recursively.

Now, we need to combine these Pareto frontiers to generate the Pareto frontier of the tree rooted at the hub. 
Before doing so, in each allocation in $\mathcal{F}_u$,
we add $u$ to the bundle of the first agent. This is a necessary step as none of the vertices connected to the hub had yet been allocated. 

Finally, we combine these $\mathcal{F}_u$s,
by looking at all possible sorted combinations
of allocations in the lists. That is, for allocations $A\in \mathcal{F}_u$ and $A'\in \mathcal{F}_{u'}$ we consider allocations $A^{\#}$ where $A^{\#}_i=A_i\cup A_j'$ for $i,j\in [n]$.
We only keep the ones that are not
weakly Pareto dominated by any other
(where an allocation weakly Pareto dominates
another allocation if it Pareto dominates it,
or all agents have equal costs in both allocations).

\begin{algorithm}[!t] 
    \caption{FindParetoFrontier($n$, $G$,$h$)}
    \label{alg:ef1+po:main}
    \begin{algorithmic}[1]{\small
        \STATE $\mathcal{F} \leftarrow [ (\emptyset, \dots, \emptyset) ]$
        \FOR{$u \in$ children of $h$}
            \STATE $T_u \leftarrow$ a subtree rooted in $u$
            \STATE $\mathcal{F}' \leftarrow $ FindParetoFrontier($n$, $T_u$, $u$)
            \STATE \textbf{for} $A \in \mathcal{F}'$ \textbf{do} add $u$ to $A_1$
            \STATE $\mathcal{F} \leftarrow$ maximal set of sorted combinations of allocations from $\mathcal{F}$ and $\mathcal{F}'$ such that none is weakly Pareto dominated by another within the set
        \ENDFOR
        \RETURN $\mathcal{F}$ }
        \end{algorithmic}
\end{algorithm}
\begin{example}{
    We run our algorithm on the instance with 2 agents from \cref{fig:motiv}.
    First, we run it on two smaller graphs:
    the vertex $a$ as one
    and one on the branch containing vertices $b,c,d,e,f$ and $g$.
    Vertex $a$ is a leaf, so 
    we get one allocation, i.e., $\mathcal{F} = \{(\{a\},\emptyset)\}$.
    When $b$ is the hub,
    we obtain two allocations:
    either one agent services $g$ along with all its ancestors
    and the other agent services $c$,
    or one agent services everything.
    Thus, $\mathcal{F}' = \{(\{b,d,e,f,g\}, \{c\}), (\{b,c, d, e,f,g\}, \emptyset)\}$.
    Finally, we combine $\mathcal{F}$ with $\mathcal{F}'$.
    We consider all four possible combinations.
    However, one of them, $(\{a,b,d,e,f,g\}, \{c\})$
    is Pareto dominated by another,
    $(\{b,c,d,e,f,g\}, \{a\})$
    (the cost of the first agent is the same,
    but for the second agent it decreases by 1).
    In conclusion, we return three allocations:
    $(\{b,d,e,f,g\}, \{a, c\})$,
    $(\{b,c, d,e,f,g\}, \{a\})$, and
    $(\{a,b, c, d,e,f,g\}, \emptyset)$.
    We note that the first one is in fact \MMS{} and \PO{} allocation. }
\end{example}
Now, we can prove the correctness of our algorithm. While we defer the full proof to \cref{app:xp}, we provide an outline here. 

\begin{restatable}{theorem}{xpalgproof}\label{thm:xpalg}
Given a delivery instance $I=\langle [n], G, h \rangle$, \cref{alg:ef1+po:main} computes its Pareto frontier and runs in time $O((n+2)! m^{3n+2})$,
where $m=|E(G)|$.
\end{restatable}

\begin{proof}[Proof (sketch)]
    Here, we note two key observations.

    For the first, consider an arbitrary instance
    $\langle [n], G, h\rangle$,
    and a smaller one obtained 
    by taking a subtree rooted in some vertex $u \in V$,
    i.e., $\langle [n], T_u,u \rangle$.
    Then, every \PO{} allocation $A$ in the original instance,
    must be still \PO{} when we cut it to the smaller instance
    (i.e., we remove vertices  not in $T_u$).
    Otherwise, a Pareto improvement on the cut allocation,
    would also be a Pareto improvement for $A$.
    Hence, by looking at all combinations of \PO{}
    allocations on the branches outgoing from $h$,
    we obtain all \PO{} allocations
    in the original instance.

    The second observation is that
    we do not need to keep two allocations
    that give the same cost for each agent.
    The maximum cost of each agent is $m$.
    Hence, we will never have more than $(m+1)^n$
    different allocations in the list
    (in fact, since we keep all allocations sorted
    in non-increasing cost order, this number
    will be much smaller).
    Thus, we can combine two frontiers efficiently.
\end{proof}

We can use \cref{alg:ef1+po:main} to obtain the desired allocations.
\begin{restatable}{theorem}{xpalgallocproof}\label{thm:xpalg:alloc}
There exists an XP algorithm parameterized by $n$,
that given a delivery instance $I=\langle [n], G, h \rangle$,
computes an \MMS{} and \PO{} allocation,
and decides whether there exist
\MMS{} and \SO{},
\EFone{} and \PO{},
and \EFone{} and \SO{} allocations.
\end{restatable}
\begin{proof}[Proof (sketch)]
    Here, we focus on \MMS{} and \PO{} allocations
    (see the \cref{app:xp} for the other notions).
    By \cref{thm:xpalg}, we can obtain a Pareto frontier
    for every instance.
    Also, from the proof of \cref{thm:xpalg},
    we know that a Pareto frontier contains at most $(m+1)^n$
    allocations.
    Thus, we can search through them to find the leximin optimal one,
    which will be \MMS{} and \PO{} as well.
\end{proof}
\noindent We note that the proof
for \EFone{} and \PO{} here
relies on \cref{thrm:ef1+po:implies_mms}.


\begin{table}[t!]
\centering
\setlength{\tabcolsep}{2pt}
\small
\begin{tabular}{@{}llll@{}}
\toprule
         & \multicolumn{1}{l}{   \textbf{\PO{}}}
         & \multicolumn{1}{l}{\textbf{Computation}\quad\quad}
         & \multicolumn{1}{l}{\textbf{Upper Bound on PoF}} \\
\midrule 
\multirow{2}{*}{\textbf{Min Cost \EFone{}}} & \xmark  &NP-h&$\frac{n(2m-n+1)}{2m}$  \\
                                   &{\footnotesize (\cref{prop:mincostexist})\quad }&\footnotesize (\cref{thm:mincostfair} ) &\footnotesize (\cref{prop:priceefone})\\
\midrule

\multirow{2}{*}{\textbf{Min Cost \MMS{}}}\quad \quad & \xmark  &NP-h&$\frac{n(m-n+1)}{m}$  \\
                                   &{\footnotesize (\cref{prop:mincostexist})\quad }&\footnotesize(\cref{thm:mincostfair} ) &\footnotesize (\cref{prop:priceefone})\\
\bottomrule  
 
\end{tabular}
\caption{The summary of our results on the price of \EFone{} and \MMS{}. 
\cmark{} denotes that the min cost allocation is always \PO{}, and \xmark{} that it may not be. 
}
\label{tab:PoF}
\end{table}

\section{Price of Fairness}\label{sec:pof}
In this section,
we study the \emph{price of fairness},
i.e., the loss to the aggregate cost
incurred by requiring allocations to be \EFone{} or \MMS{} when delivery tasks are located on the vertices of a tree.
Specifically, we show tight bounds
for the ratio of the aggregate cost of a fair allocations to the cost of an \SO{} allocation.
The price of fairness concept has been well studied in fair division literature
\citep{barman2020optimal,sun2021fairness,bhaskar2023price} and merits exploration in delivery settings as well.

Formally, price of fairness is defined 
as the ratio of the minimum total cost of an allocation satisfying the given notion
to the minimum total cost of any allocation. 

\begin{definition}[Price of Fairness]
    Given a fairness concept $F$ and an instance $I=\langle [n], G, h \rangle$, the price of $F$ is given as:
    \[
    \mathrm{PoF}(I)=  \frac{\min_{A \in \Pi^n: A \text{ satisfies } F} \sum_{i \in [n]}c(A_i)}{ \min_{A \in \Pi^n} \sum_{i \in [n]}c(A_i)} 
    \]
\end{definition}

We first remark on the efficiency of minimum cost fair solutions.

\begin{restatable}{proposition}{mincostexist}\label{prop:mincostexist}
    Given a delivery instance $I=\langle [n],G, h\rangle$, we have that
    \begin{itemize}
        \item A minimum cost \MMS{} allocation must be \PO{}.
        \item A minimum cost \EFone{} allocation need not be \PO{}
    \end{itemize}
\end{restatable}

\begin{proof}
    We first consider minimum cost \MMS{} allocations. Given a delivery instance $I=\langle [n], G, h \rangle$, let $A=(A_1,\cdots, A_n)$ be a minimum cost \MMS{} allocation. We shall show that it must be \PO{}. 
    
    For contradiction, let $B=(B_1,\cdots, B_n)$ be an allocation that Pareto dominates $A$. That is, for all $i\in [n]$, $c(B_i)\leq c(A_i)$ and there exists $i^*\in [n]$ s.t. $c(B_{i^*})<c(A_{i^*})$. As $A$ is \MMS{}, this implies that $B$ must also be \MMS{}. Further, due to Pareto domination, $B$ must have strictly lower social cost than $A$. Clearly, this contradicts the fact that $A$ is  a minimum cost \MMS{} allocation. Consequently, for any delivery instance, a minimum cost \MMS{} allocation must be \PO{}.

    The same is not true for \EFone{}
    as an \EFone{} and \PO{} allocation
    need not to exist (\cref{prop:non-existenceEFonePO}) but minimum cost \EFone{} allocations always exist.
\end{proof}

We now show that it is intractable to find minimum cost fair allocations.
\begin{restatable}{theorem}{mincostfair}\label{thm:mincostfair}
    Given a delivery instance $I=\langle [n],G, h\rangle$, computing a minimum cost \MMS{} allocation or a minimum cost \EFone{} allocation is \NPH{}.
\end{restatable}
\begin{proof}
    We first consider minimum cost \MMS{} allocations. Clearly, whenever an \MMS{} and \SO{} allocation exists, it will also be a minimum cost \MMS{} allocation. Thus, an algorithm that always finds a minimum cost \MMS{} allocation would find an \MMS{} and \SO{} allocation, whenever it exists. Thus, finding a minimum cost \MMS{} allocation is \NPH{}. 
    
    For the case of minimum cost \EFone{} allocations, recall the reduction in \cref{prop:po:hardness}. In the instance constructed, any \EFone{} and \PO{} allocation will be \SO{}. Consequently, it would be a minimum cost \EFone{} allocation. Hence, finding a minimum cost \EFone{} allocation is \NPH{}.
\end{proof}

In the following proposition we
identify the tight bounds for
the price of \EFone{} and \MMS{}.
\begin{restatable}{proposition}{priceefone}\label{prop:priceefone}
    Given a delivery instance $I=\langle [n],G, h\rangle$, 
    it holds that
    $\mathrm{Po}\MMS{}(I)\leq \frac{n(m-n+1)}{m}$ and
    $\mathrm{Po}\EFone{}(I)\leq \frac{n(2m-n+1)}{2m}$.
    Both these bounds are tight.
\end{restatable}

\begin{proof}
    In delivery instances, the minimum total cost of any allocation, fair or otherwise, will be achieved by socially optimal allocations that have social cost $m$. Thus, for price of fairness in delivery instances, the denominator of the fraction will always be $m$.
    We first consider upper bounds on the price of \MMS{}.

\begin{figure}[!t]
    \centering
    \begin{tikzpicture}
        \def\xs{1cm} 
        \def\ys{0.5cm} 
        \def\x{0cm} 
        \def\y{0cm} 
        \def\ls{\footnotesize} 
        
        \tikzset{
            node_blank/.style={circle,draw,minimum size=0.55cm,inner sep=0, color=white}, 
            node/.style={circle,draw,minimum size=0.4cm,inner sep=0, fill = black!05},
            node_h/.style={circle,draw,minimum size=0.65cm,inner sep=0, fill = blue!20, font=\footnotesize},
            node_emph/.style={circle, minimum size=0.8cm, black!15, fill = black!15, draw,inner sep=0, font=\footnotesize},
            edge/.style={draw,thick,sloped,-,above,font=\footnotesize},
            arrow/.style={draw, single arrow, minimum width = 0.9cm, minimum height=\y-6*\x+\s, fill=black!10},
            blank/.style={},
            alloc_a/.style={circle, draw, color=blue!30, minimum size=0.8cm, line width=2.5},
            alloc_b/.style={rectangle, draw, color=red!30, minimum size=0.7cm, line width=2.5},
        }
        
        \node[node_emph] (_) at (\x + 1*\xs, 1*\ys + \y) {};
        
        \node[node_h] (h) at (\x + 1*\xs, 1*\ys + \y) {\ls $h$};
        \node[node] (a) at (\x + 2*\xs, 1*\ys + \y) {};
        \node[node] (b) at (\x + 3*\xs, 1*\ys + \y) {};
        \node[node] (c) at (\x + 4*\xs, 1*\ys + \y) {};
        \node[node] (d) at (\x + 5*\xs, 1*\ys + \y) {};
        \node[node] (e) at (\x + 6*\xs, 1*\ys + \y) {};
        \node[node] (f) at (\x + 7.1*\xs, 0*\ys + \y) {};
        \node[node] (g) at (\x + 7.1*\xs, 2*\ys + \y) {};

        \node[alloc_a] (_) at (\x + 2*\xs, 1*\ys + \y) {};
        \node[alloc_a] (_) at (\x + 3*\xs, 1*\ys + \y) {};
        \node[alloc_a] (_) at (\x + 4*\xs, 1*\ys + \y) {};
        \node[alloc_a] (_) at (\x + 5*\xs, 1*\ys + \y) {};
        \node[alloc_a] (_) at (\x + 6*\xs, 1*\ys + \y) {};
        \node[alloc_a] (_) at (\x + 7.1*\xs, 0*\ys + \y) {};
        \node[alloc_b] (_) at (\x + 7.1*\xs, 2*\ys + \y) {};
        
        \path[edge]
        (h) edge (a)
        (a) edge (b)
        (b) edge (c)
        (c) edge (d)
        (d) edge (e)
        (e) edge (f)
        (e) edge (g)
        ;
      
    \end{tikzpicture}
    \caption{An example with $m=7$ and $n=2$. A minimum cost \MMS{} allocation is depicted. Here $\mathrm{PoMMS}=\frac{12}{7}=\frac{2 \cdot (7-2+1))}{7}$}
    \label{fig:pommstight}
\end{figure}

    \paragraph{Price of MMS.}  Let $I = \langle [n], G, h \rangle$ be an arbitrary delivery instance and
    $A$ a minimum cost \MMS{} allocation. In any \MMS{} allocation, each agent has cost at most the \MMS{} threshold. In the worst case, the number of vertices that are visited by multiple agents are maximized. We have from \cref{prop:mincostexist} that a minimum cost \MMS{} allocation must be \PO{}. In a \PO{} allocation, each agent must service at least one leaf.
    Otherwise, all of the vertices serviced by it
    can be given to the agents serving the descendants of these vertices
    decreasing the cost for this agent to 0 and
    incurring no additional cost on any other agent,
    which would be a Pareto improvement.

    In the worst case, all non leaf vertices are visited by all agents and each agent additionally services a single leaf. This would mean that there are $n$ leaves, and each agent visits the $m-n$ internal vertices and $1$ leaf. This gives a social cost of $n(m-n+1)$.
    Thus, we have that
    \[\mathrm{Po}\MMS{}(I)\leq \frac{n(m-n+1)}{m} \]

    To show that this bound is tight for every $m$ and $n$,
    we consider the graph that is the hub and $n$-armed star connected by a path of length $m-n$ (see \cref{fig:pommstight} for an illustration).
    Formally, 
    \[
    G = (\{h, v_1, \dots, v_m\},
        \{(h,v_1),(v_1,v_2),\dots,(v_{m-n-1},v_{m-n}),
        (v_{m-n},v_{m-n+1}),\dots,(v_{m-n},v_m)\}).
    \]
    Observe that in an \MMS{} allocation, each vertex $v_{m-n+1},\dots,v_{m}$ has to be served by a different agent
    (otherwise the cost of an agent serving two or more of such vertices would be greater than $m-n+1$, which is the value of the \MMS{}).
    Thus, the cost of every \MMS{} allocation matches the established upper bound.

    \paragraph{Price of EF1.}  Recall that the price of \EFone{} is the ratio of the social cost of the minimum cost \EFone{} allocation and $m$. In trying to give an upper bound on this, we again aim to maximize the number of vertices that are visited by multiple agents, conditioned on the allocation having minimum cost over all \EFone{} allocations.

    Let $I = \langle [n], G, h \rangle$
    be an arbitrary instance and
    $A$ an arbitrary minimum cost \EFone{} allocation.
    We shall prove by induction
    that for every $k \in [n]$
    there exists a set of $k$ vertices, $S_k$,
    such that $V \setminus S_k$ induces a connected graph
    and the vertices in $S_k$ are
    visited by at most $k$ agents,
    i.e., $|\{ i \in [n] : V(G| A_i) \cap S_k \neq \emptyset\}| \le k$.
    For $k=1$, observe that
    there has to exist at least one leaf in graph $G$
    and the leaf is visited only by the agent
    that serves it.
    Also, removing this leaf retains a connected graph.
    Hence, the induction basis holds.
    
    Now, assume that the induction hypothesis holds for some $k \in [n-1]$ and consider $k+1$.
    We know that there exists a set $S_k$ such that
    $|S_k| = k$ and graph induced by $V \setminus S$, i.e., $G' = G[V \setminus S]$ is connected.
    Let $v$ be an arbitrary leaf in $G'$.
    Observe that apart from the agent serving $v$
    it can be visited only
    by agents serving vertices in $S$.
    Hence, $v$ is visited by at most $k+1$ agents.
    Then, $S \cup \{v\}$ contains $k+1$ vertices
    visited by at most $k+1$ agents.
    Observe also that removing this set of vertices
    retains a connected graph.
    Hence, by induction, the thesis holds.


    Consequently,
    we get that there is an agent not visiting $n-1$ vertices in $S_{n-1}$,
    an agent not visiting $n-2$ vertices in $S_{n-2}$ and so on.
    This yields the following upper bound on
    the total cost in minimum cost \EFone{}:
    \begin{align*}
        \mathrm{Po}\EFone{} &\leq
        \frac{n \cdot m - \sum_{k \in [n]} (k-1) }{m} =
        \frac{n \cdot m - n(n-1)/2}{m} =
        \frac{n \cdot (2m - n + 1)}{2m}.
    \end{align*}

    For an example of an instance where this bound is tight for every $m$ and $n$, we consider a path graph with hub at one of the ends of the path (see \cref{fig:poefonetight} for an illustration).
    Formally, let $G=(\{h,v_1,\dots,v_m\},\{(h,v_1),(v_1,v_2),\dots,(v_{m-1},v_m)\})$.
    Observe that in order for the allocation to be \EFone{},
    each vertex $v_{m}, v_{m-1}, \dots, v_{m-n+1}$ has to be served by a different agent.
    This yields the cost equal to the established upper bound.
\end{proof}

    \begin{figure}[!t]
    \centering
    \begin{tikzpicture}
        \def\xs{1cm} 
        \def\ys{0.4cm} 
        \def\x{0cm} 
        \def\y{0cm} 
        \def\ls{\footnotesize} 
        
        \tikzset{
            node_blank/.style={circle,draw,minimum size=0.55cm,inner sep=0, color=white}, 
            node/.style={circle,draw,minimum size=0.4cm,inner sep=0, fill = black!05},
            node_h/.style={circle,draw,minimum size=0.65cm,inner sep=0, fill = blue!20, font=\footnotesize},
            node_emph/.style={circle, minimum size=0.8cm, black!15, fill = black!15, draw,inner sep=0, font=\footnotesize},
            edge/.style={draw,thick,sloped,-,above,font=\footnotesize},
            arrow/.style={draw, single arrow, minimum width = 0.9cm, minimum height=\y-6*\x+\s, fill=black!10},
            blank/.style={},
            alloc_a/.style={circle, draw, color=blue!30, minimum size=0.8cm, line width=2.5},
            alloc_b/.style={rectangle, draw, color=red!30, minimum size=0.7cm, line width=2.5},
        }
        
        \node[node_emph] (_) at (\x + 1*\xs, 1*\ys + \y) {};
        
        \node[node_h] (h) at (\x + 1*\xs, 1*\ys + \y) {\ls $h$};
        \node[node] (a) at (\x + 2*\xs, 1*\ys + \y) {};
        \node[node] (b) at (\x + 3*\xs, 1*\ys + \y) {};
        \node[node] (c) at (\x + 4*\xs, 1*\ys + \y) {};
        \node[node] (d) at (\x + 5*\xs, 1*\ys + \y) {};
        \node[node] (e) at (\x + 6*\xs, 1*\ys + \y) {};

        \node[alloc_a] (_) at (\x + 2*\xs, 1*\ys + \y) {};
        \node[alloc_a] (_) at (\x + 3*\xs, 1*\ys + \y) {};
        \node[alloc_a] (_) at (\x + 4*\xs, 1*\ys + \y) {};
        \node[alloc_a] (_) at (\x + 5*\xs, 1*\ys + \y) {};
        \node[alloc_b] (_) at (\x + 6*\xs, 1*\ys + \y) {};

        \path[edge]
        
        (h) edge (a)
        (a) edge (b)
        (b) edge (c)
        (c) edge (d)
        (d) edge (e)
        
        ;
      
    \end{tikzpicture}
    \caption{An example with $m=5$ and $n=2$. A minimum cost \EFone{} allocation is depicted. Here $\mathrm{PoEF1}=\frac{9}{5}=\frac{2 \cdot (10 -2 + 1)}{10}$.}
    \label{fig:poefonetight}
\end{figure}

We note that the bound for the price of \MMS{} is lower,
which matches our intuition that \MMS{} is more compatible with efficiency in our setting.
Recall from \cref{thm:mincostfair} that finding a min cost \MMS{} or \EFone{} allocation is \NPH{}.
Further, from  \cref{prop:mincostexist}, a min cost \MMS{} allocation has to be \PO{}
(which is not the case for \EFone{}).
This allows us to use \cref{alg:ef1+po:main}
to calculate the price of \MMS{}
and study it in numerical experiments in \cref{sec:exp}.



\section{Experimental Results}\label{sec:exp}

We now present our experimental results on the running time of our algorithm, the existence of fair and efficient allocations and investigate the efficiency loss of fair solutions through their \textit{price of fairness}. 
In each experiment, we generated trees, uniformly at random, based on Prüfer sequences~\citep{prufer27sequences}
using NetworkX Python library~\citep{hagberg2008networkx}.%
\footnote{The code for our experiments is available at:
\url{https://doi.org/10.5281/zenodo.11149658}.}
For each experiment and a graph size, we sampled 1,000 trees.

We note that the experiments were conducted on
Dell Latitude E5570 computer
with Intel(R) Core(TM) i7-6600U CPU
@ 2.60GHz 2.80 GHz processor,
16.0 GB of RAM, and
Windows 10 Education operating system.
We used Python version 3.10.7
with NetworkX library version 2.8.8
(the plots were created using MatPlotLib v. 3.6.2).

\subsection[Running Time of Algorithm 1]{Running Time of \cref{alg:ef1+po:main}}
In our first experiment, we observe the running time of our algorithm to find the Pareto frontier of a given delivery instance. 
We run \cref{alg:ef1+po:main} for graphs of sizes $10,20,\dots,100$
and every number of agents from $2$ to $6$.
The running times are reported in \cref{fig:alg:running-time}
(the running time for $2$ agents is not reported
as it would be indistinguishable
from the running times for 3 agents
in the picture).
The power in the running time of our algorithm
for $n=6$ agents is significant,
hence the sharp increase in the running time
with growing graph size in this case.

\begin{figure*}[!t]
    \centering
    \begin{subfigure}{0.46\textwidth}
        \centering
        \includegraphics[width=0.8\linewidth]{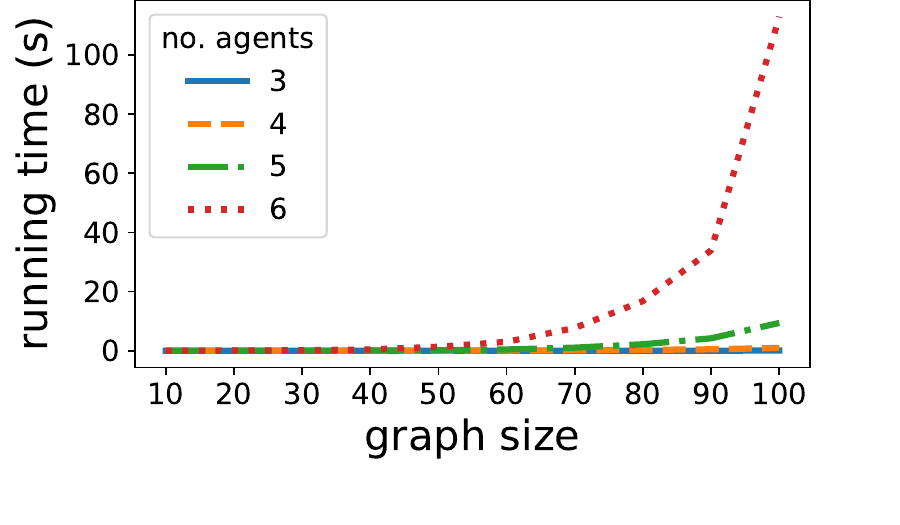}
        \caption{Running time.}
        \label{fig:alg:running-time}
    \end{subfigure}
    \begin{subfigure}{0.39\textwidth}
        \centering
        \includegraphics[width=\textwidth]{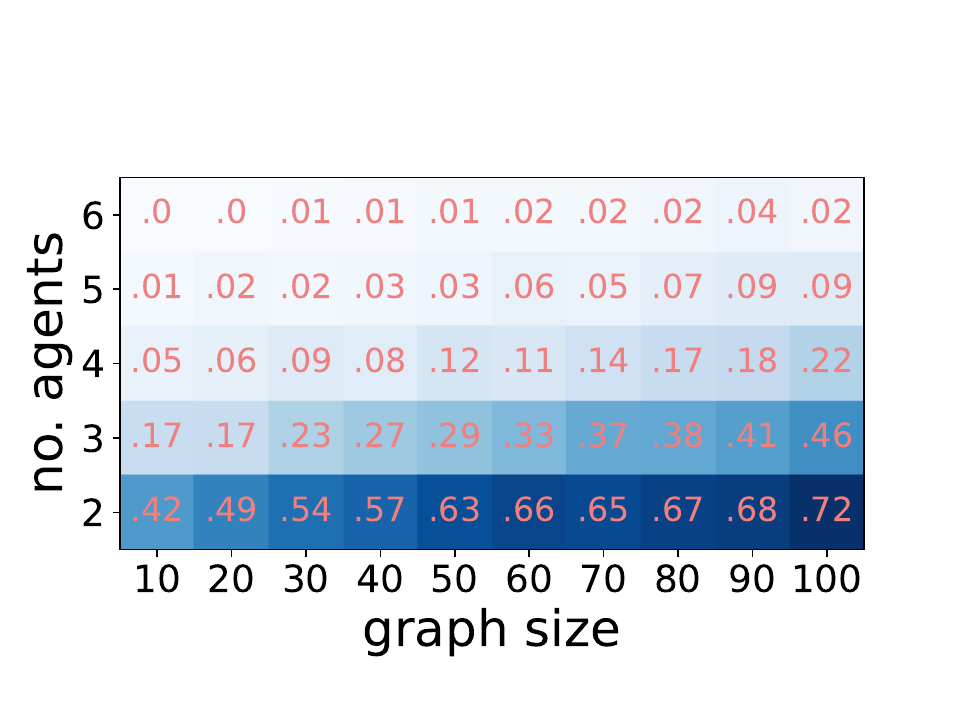}
        \caption{\EFone{} and \PO{}  existence.}
        \label{fig:exp:ef1+po:prob}
    \end{subfigure}
    \hfill
    \begin{subfigure}{0.46\textwidth}
        \centering
        \includegraphics[width=0.8\linewidth]{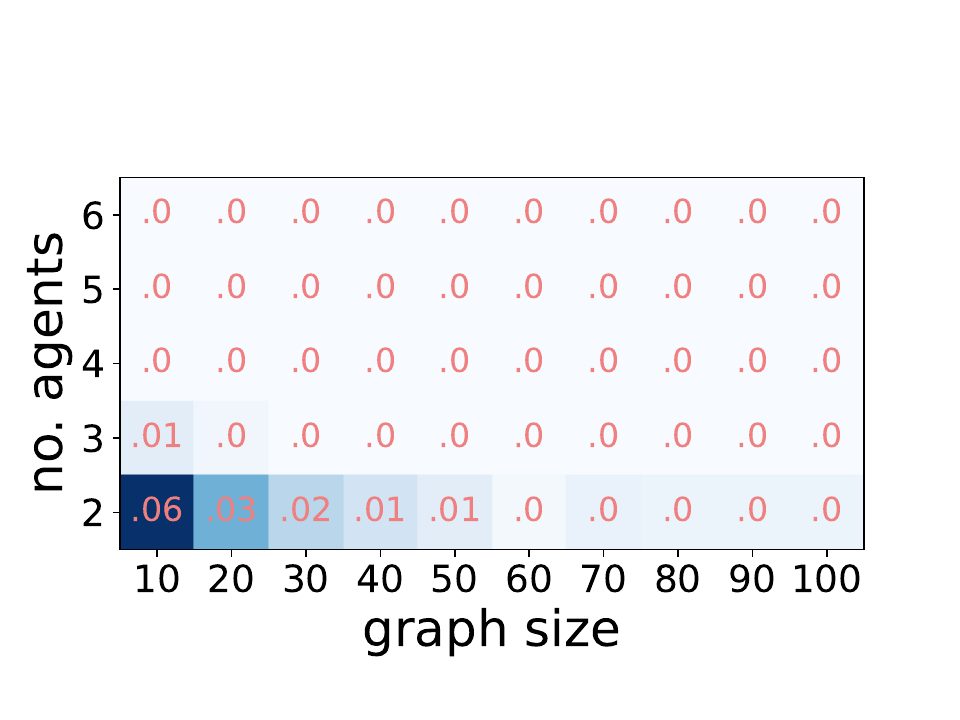}
        \caption{\EFone{} and \SO{} allocation existence.}
        \label{fig:mms+so:prob}
    \end{subfigure}
    \begin{subfigure}{0.46\textwidth}
        \centering
        \includegraphics[width=0.8\linewidth]{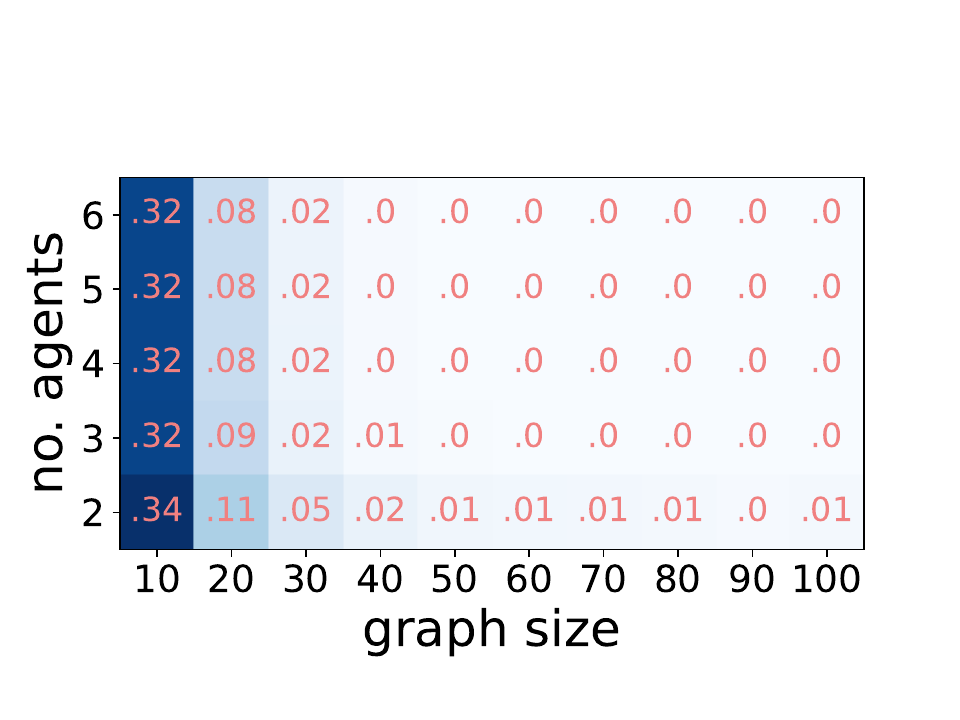}
        \caption{\MMS{} and \SO{} allocation existence.}
        \label{fig:ef1+so:prob}
    \end{subfigure}
    \caption{ \cref{fig:alg:running-time} illustrates the average running time of \cref{alg:ef1+po:main} on graphs of different sizes with different number of agents. \cref{fig:exp:ef1+po:prob,fig:mms+so:prob,fig:ef1+so:prob} show the fraction
    of instances admitting \EFone{} and \PO{}, \EFone{} and \SO{} and \MMS{} and \SO{} allocation, respectively, for different number of agents and sizes of graphs.}
\end{figure*}

\subsection{Existence of Fair and Efficient Allocations}

In this experiment, we check how often fair and efficient allocations exist for randomly generated trees. As we have seen, \MMS{} and \PO{} allocations always exist. Consequently, we focus on the three other combinations of fairness and efficiency. 

First, we checked how often there exists an \EFone{} and \PO{} allocation.
To this end, we generated trees of sizes $10, 20, \dots, 100$
and for each tree we run \cref{alg:ef1+po:main} for each
number of agents from 2 to 6.
Based on the output, we checked
the number of trees that admit an \EFone{} and \PO{} allocation.
As shown in \cref{fig:exp:ef1+po:prob}, the probability of finding an \EFone{} and \PO{} allocation
increases steadily when we increase the size of the graph,
but drops sharply when we increase the number of agents.
Intuitively, on larger graphs we have more flexibility in how we fairly and efficiently split
the vertices.
However, when there are more agents,
it may still be difficult to satisfy fairness for each of them.
We repeat a similar experiment for \EFone{} and \SO{}
as well as \MMS{} and \SO{} allocations.

The results for \EFone{} and \SO{} allocations
are presented in \cref{fig:mms+so:prob}.
There, we see a sharp decrease in probability
with the increase in either number of agents
or the size of a graph.
For the former it is clear,
as with larger number of agents
it is difficult to be fair to all of them.
For the size of a graph,
recall \cref{prop:center:ef1so}
in which we have shown that
\EFone{} and \SO{} allocation exists
only if  the hub is in the center of a tree.
Moreover, the center of a tree always contains one or two vertices.
Hence, with the increase in the size,
the probability that the hub will be in the center decreases.

The results for \MMS{} and \SO{}
allocations are presented in \cref{fig:mms+so:prob}.
As can be seen in the picture,
the probability does not really vary much
between different numbers of agents
(especially for small graphs).
A plausible explanation for that phenomenon
is that for \MMS{} we only have to care to not be
unfair to the worst off agent.
If we have a small graph and a lot of agents,
then probably some of them will not be assigned
to any vertex either way.
However,  the number of such agents
does not impact whether an allocation is \MMS{} or not.
Hence the visible effect.

\begin{figure}[t]
    \centering
    \begin{subfigure}{0.48\textwidth}
        \includegraphics[width=\linewidth]{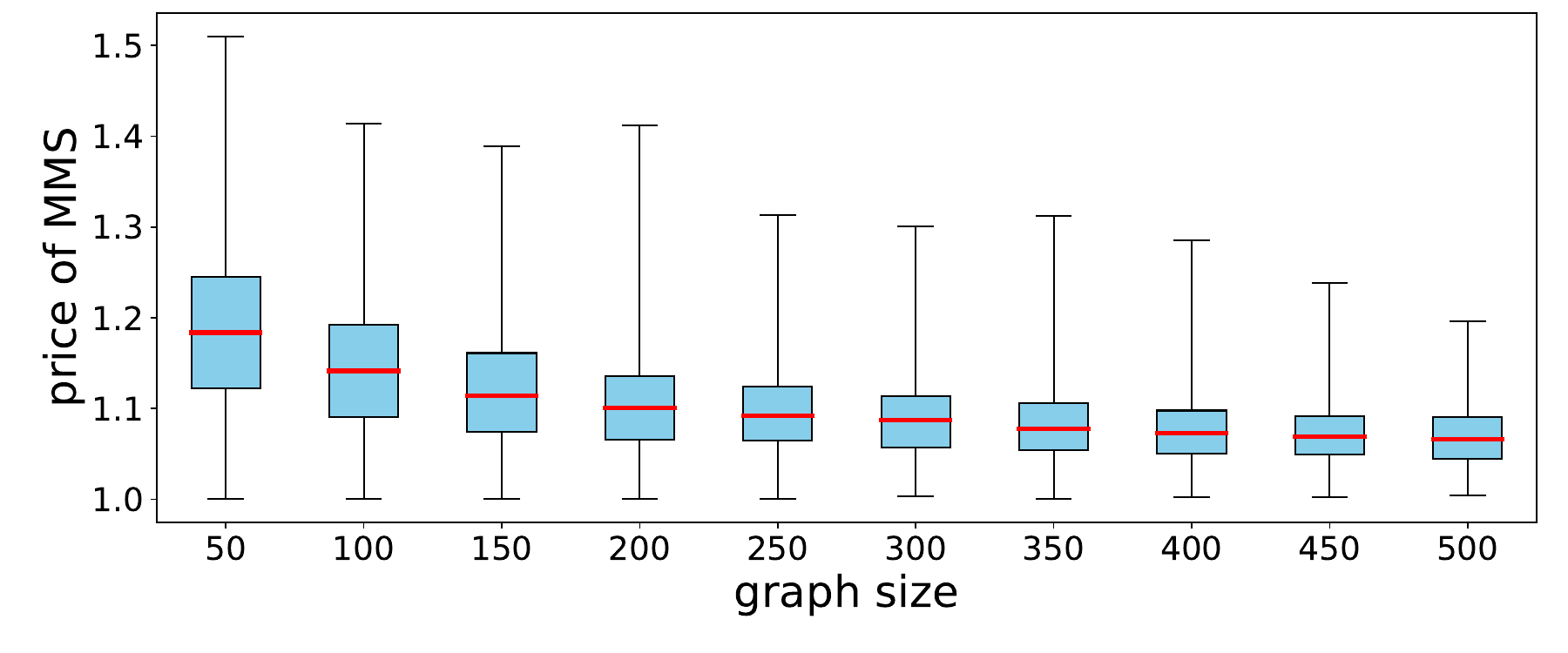}
        \caption{Price of \MMS{} for 2 agents.}
        \label{fig:price_of_mms:app:2}
    \end{subfigure}
    \begin{subfigure}{0.48\textwidth}
        \includegraphics[width=\linewidth]{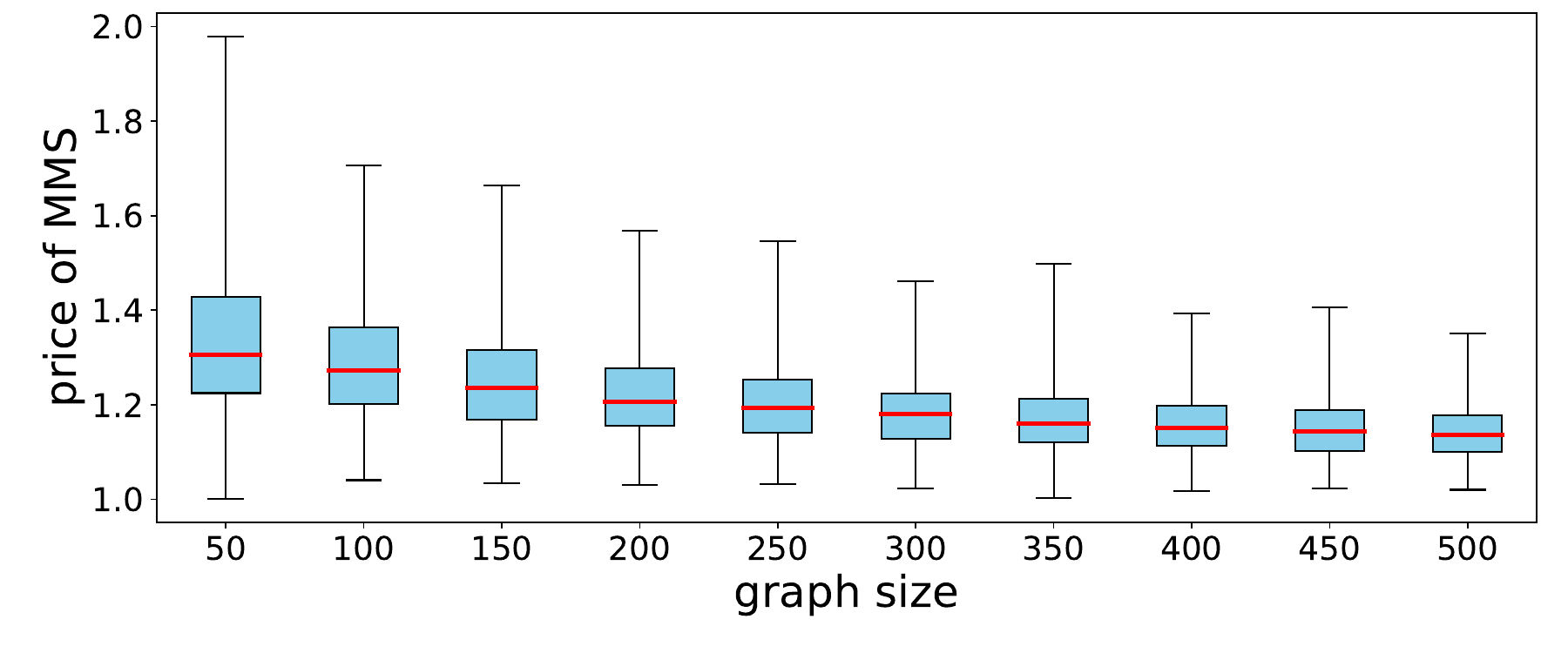}
        \caption{Price of \MMS{} for 3 agents.}
        \label{fig:price_of_mms:app:3}
    \end{subfigure}
    \caption{Price of \MMS{} for graphs of different sizes afor 2 and 3 agents.}
\end{figure}

\subsection{Price of \MMS{}}
Next, we analyze the price of \MMS{} observed in randomly generated trees. Recall that the price of fairness is the ratio of the sum of the agents' costs under a min cost fair allocation to the minimum possible social costs. From \cref{prop:mincostexist}, a minimum cost \MMS{} allocation must be \PO{}. However, an \EFone{} and \PO{} allocation need not exist. Furthermore, we do not have an exact algorithm for finding a minimum cost \EFone{}. Recall that this problem is \NPH{}, as shown in \cref{thm:mincostfair}. Consequently, we only pursue the price of \MMS{} experimentally. 

In this experiment, we observe the price of \MMS{} randomly generated trees of different sizes. 
In particular, we observe the change in the median price of fairness observed for each tree size. To this end, we generated trees of sizes $50,100,150,\dots,500$
and looked for a min cost \MMS{} allocation for 2 and 3 agents using \cref{alg:ef1+po:main}.
The results are reported in boxplots in
\cref{fig:price_of_mms:app:2,fig:price_of_mms:app:3}.
For 2 agents, \cref{fig:price_of_mms:app:2} illustrates that the median price is around $1.19$ for the graphs of size 50
and it steadily decreases for the larger graphs, to $1.05$ for graphs of size 500.
For 3 agents, \cref{fig:price_of_mms:app:3} illustrates that the median price is around $1.3$ for the graphs of size 50
and here too, it steadily decreases for the larger graphs, to $1.17$ for graphs of size 500.
These results suggest that as the size of the instance grows, the efficiency loss due to \MMS{} becomes negligible in most cases
(at least for a small number of agents). For a fixed graph size, the price of \MMS{} seems to be higher for more agents than fewer, which is intuitive.

\subsection{Pareto Frontiers}
In our final experiment, we analyze the trajectories of Pareto frontiers 
in randomly generated instances of fixed sizes
in order to establish an empirical 
trade-off between fairness (the maximum difference in the costs)
and efficiency (the total cost of agents).

We focus on graphs of size 100 or 400,
with 2-4 agents (for 4 agents, we only consider graphs of size 100). 
For each size and the number of agents and
for each of the 1000 trees generated,
we look at each allocation in the Pareto frontier 
and report the total cost of all agents on y-axis
and the maximum difference in the costs of the agents on x-axis
(in case of multiple allocations in the frontier
with the same total cost, we report only the one in which 
the maximum difference in the costs is the smallest).
Then, we connect all such points for one Pareto frontier
to form a partially transparent blue line.
By superimposition of all 1000 of such blue lines,
we obtain a general view on the distribution of Pareto frontiers.
With the thick red line we denote the average total cost,
for each difference in costs. 
We see that particular Pareto frontiers can behave very differently,
but the general pattern is quite strong:

\paragraph{Two agents.}
For 2 agents and graphs of size 400,
shown in \cref{fig:po_frontier} the total cost does not vary much when
the difference in costs is between 0 and 250,
however it is much steeper for the larger differences.
These findings imply that it is usually not effective to
focus on partial fairness as
the additional total cost that we incur
by guarantying complete fairness instead of partial
is not that big.
The plot for the size 100 and 2 agents, shown in \cref{fig:po:front:2:100} looks very similar.
Again, we can say, that the total cost for agents increases
sharply when we decrease the difference in costs of agents from 100,
but the further we go, this increase is slower.

\begin{figure}[!t]
    \centering
    \begin{subfigure}{0.32\textwidth}
        \includegraphics[width=\linewidth]{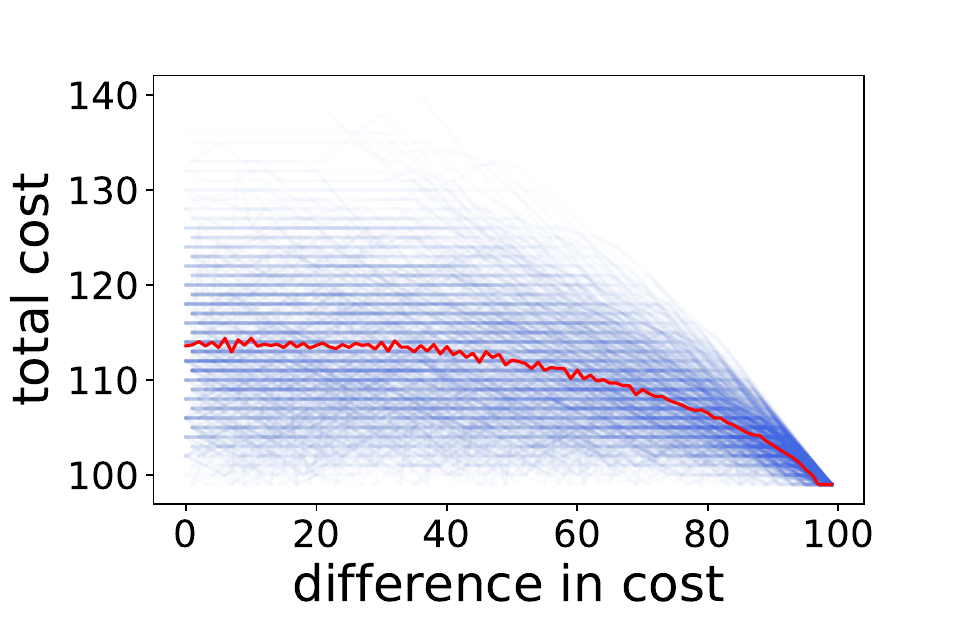}
        \caption{Pareto frontier (2,100).}
        \label{fig:po:front:2:100}
    \end{subfigure}
    \begin{subfigure}{0.32\textwidth}
        \centering
        \includegraphics[width=\textwidth]{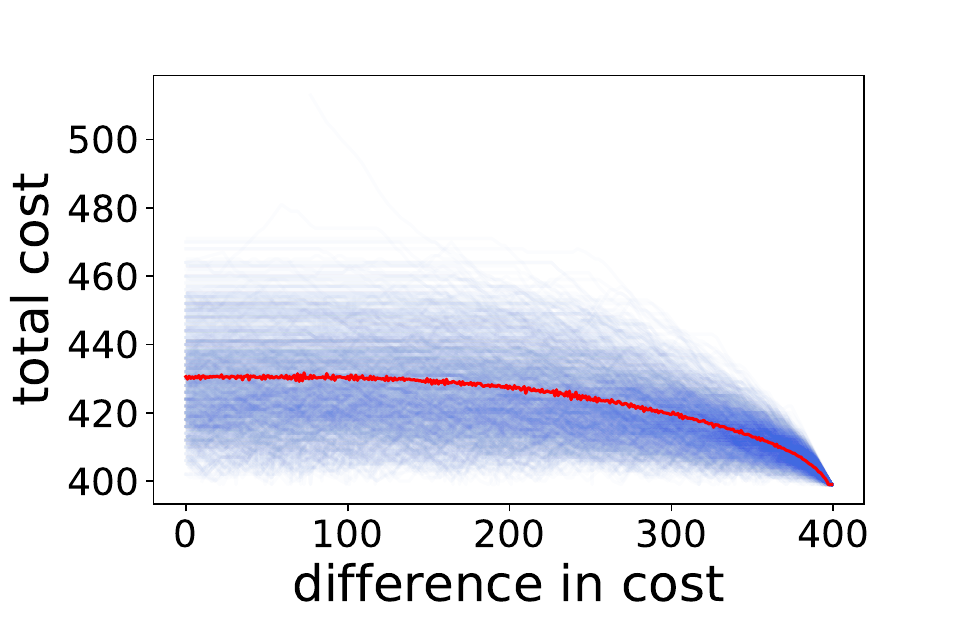}
        \caption{Pareto frontier (2,400).}
        \label{fig:po_frontier}
    \end{subfigure}  
    \begin{subfigure}{0.32\textwidth}
        \includegraphics[width=\linewidth]{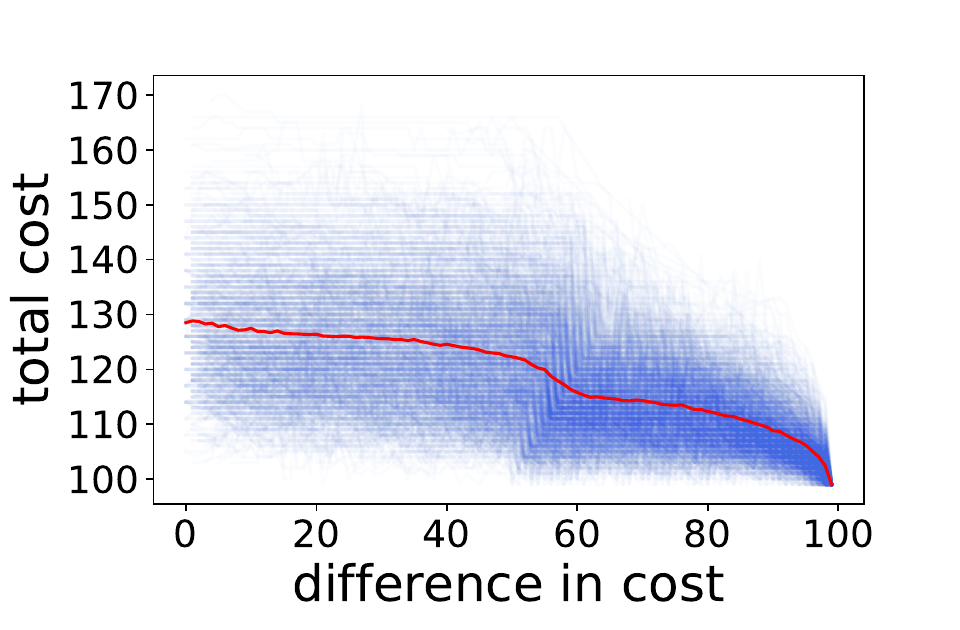}
        \caption{Pareto frontier (3,100).}
        \label{fig:po:front:3:100}
    \end{subfigure}
    \begin{subfigure}{0.32\textwidth}
        \includegraphics[width=\linewidth]{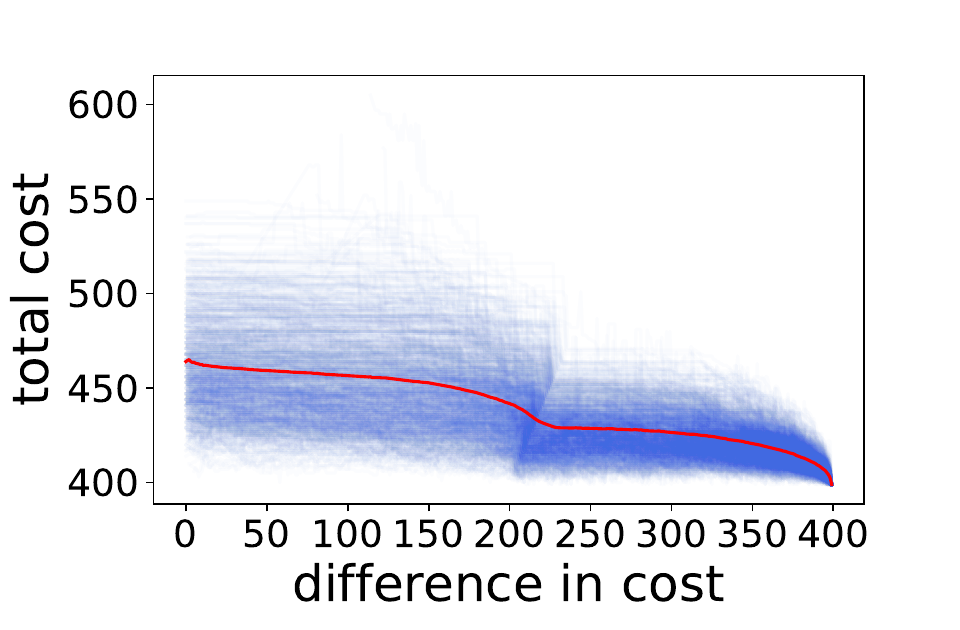}
        \caption{Pareto frontier (3,400).}
        \label{fig:po:front:3:400}
    \end{subfigure}
    \begin{subfigure}{0.32\textwidth}
        \includegraphics[width=\linewidth]{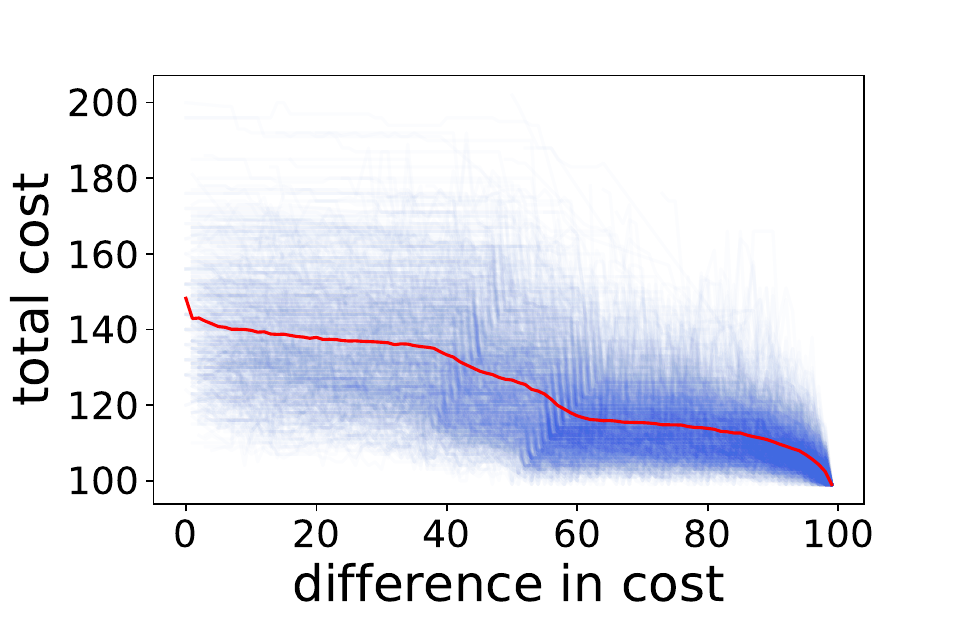}
        \caption{Pareto frontier (4,100).}
        \label{fig:po:front:4:100}
    \end{subfigure}
   
    \caption{ ``Pareto frontier ($n$, $m$)'',
    shows the distribution of Pareto frontiers for $n$ agents and graphs of size $m$.}
    \label{fig:experiments:app}
\end{figure}

\paragraph{Three agents.} We now look at the Pareto frontiers with 3 agents, shown in \cref{fig:po:front:3:100} and \cref{fig:po:front:3:400}.
For both graph sizes, the plots are not as smooth as the case for 2 agents. Specifically, there is a sharp change in the total cost near the middle of the plot (where the $x$-axis values are between 55 and 60 for \cref{fig:po:front:3:100} and between 210 and 240 for \cref{fig:po:front:3:400}).
We offer the following interpretation of this fact:
in the rightmost part of the plot we begin with
an allocation where the first agent is serving all of the vertices,
which gives us the minimum total cost,
but also the maximum difference between the costs of two agents.
Then, as we move towards the left, in the direction of decreasing differences in the agents' costs,
some of the vertices served by the first agent are now given to some other agent.

Initially, we do not observe a dramatic increase in the total costs as the difference in costs decreases. In the allocations in this area, 
the graph seems to be serviced only between the first two agents,
which  results in a smaller total cost for a higher difference in costs.
However, when we reach the difference in cost that is 
around the half of the graph size, this is no longer possible.
To decrease the difference further, 
vertices must be assigned to the third agent, which sharply increases the total cost.
This leads to the sharp increase seen in so many of the individual Pareto frontiers.

\paragraph{Four agents.} Turning to the case of four agents, shown in \cref{fig:po:front:4:100}, we observe two sharp change in the Pareto frontiers as opposed to a single change for 3 agents
(one for the $x$-axis values between 35 and 45 and another between 55 and 60).
We believe that the explanation for 4 agents is similar to that of 3 agents.
Again moving from the right most point with minimum total cost and maximum difference in cost, towards the left, we see a small increase in the total cost as at first only two agents are servicing the vertices. Then around the middle (half the number of vertices in the graph), we see a sharp increase in total costs, as now three agents are servicing the whole graph. Another sharp increase is seen when the fourth agent also starts delivering items. 



\section{Conclusions and Future Work}

We introduced a novel problem of fair distribution of delivery orders on graphs. We showed that even when  for unweighted tree graphs, the problem of finding fair and efficient solutions proves to be pretty hard. 
We provided a comprehensive characterization of the space of instances that admit fair (\EFone{} or \MMS{}) and efficient (\SO{} or \PO{}) allocations and---despite proving their hardness---developed an XP algorithm parameterized by the number of agents for each combination of fairness-efficiency notions. We found that between \MMS{} and \EFone{}, \MMS{} is more compatible with efficiency: an \MMS{} and \PO{} allocation always exists and the worst case price of \MMS{} is lower than that of \EFone{}.

Since a preliminary version of this work was published \cite{hosseini2023fair}, there has been significant interest in this model and its extensions. Some subsequent work is already available online. \cite{hosseini2025algorithmic} extends our model to weighted graphs and pursues \MMS{} and a relaxation of efficiency for different types of trees. One notable result there shows that in fact no FPT algorithm can find MMS and efficient allocations, making our XP algorithm the optimal. A different follow up considers fairness in multi stage deliveries \cite{vibulan2025fair}. Specifically, this is an offline model where delivery tasks are known some time in advance and fairness needs to be maintained throughout.

Even beyond these recent follow ups, our work paves the way for future research on developing approximation schemes
or perhaps algorithms parameterized by graph characteristics
(e.g., maximum degree or diameter) in this domain.
Another natural direction is extending our results to cyclic graphs. While the current impossibility results will continue to hold, multiplicative approximations of \MMS{} would be an important question to pursue in this space.

An alternate way to extend the current model would be by allowing agents to have heterogeneous cost functions, instead of being identical. Similarly, it would make sense to consider a setting where agents have capacity constraints on how many orders they may be allocated.
Another extension of the model could involve multiple hubs.

It would also be useful to study delivery tasks settings where tasks arrive in an online fashion. Alternatively, the underlying graph itself may be dynamic and some branches may be added or removed across time. 
Note that any model that is more general will still experience the same intractability hurdles as our model, and thus, the problem will be quite non-trivial.

Another avenue for future work would be to do the experiments we have done for price of fairness for alternate objectives like minimum cost \EFone{}. While the problem is \NPH{}, an exponential time exact algorithm, which is yet unknown, would allow such experiments. In a similar vein, it would be useful to run the price of fairness experiments on real world transportation networks.  



\section*{Acknowledgments}
Hadi Hosseini acknowledges support from NSF IIS grants \#2144413, \#2052488, and \#2107173. Part of this work was done when Shivika Narang was a visitor at Pennsylvania State University and Tomasz Wąs a postdoctoral scholar therein. Shivika Narang is currently supported by the NSF-CSIRO project on “Fair Sequential Collective Decision-Making”. We thank the anonymous reviewers for their constructive suggestions.

{\small 
\bibliographystyle{named} 
\bibliography{references.bib}}
\cleardoublepage

\appendix

\section*{Appendix}
\section{Additional Related Work}\label{app:priorwork}
Fair division of indivisible items has garnered much attention in recent years. Several notions of fairness have been explored here with \EFone{} \citep{lipton-envy-graph,budish2011combinatorial,barman2018finding,caragiannis2019unreasonable,barman2019fair} and \MMS{} \citep{barman2019fair, barman2020approximation,ghodsi2018fair, procaccia2014fair,hosseini2021guaranteeing}  being among the most prominent ones. A significant amount of work has gone into studying the coexistence of these fairness notions with efficiency, through the lens of pareto optimality \citep{caragiannis2019unreasonable,barman2018finding,hosseini2022fairly} or that of social welfare \citep{barman2019fair,caragiannis2022repeatedly,aziz2023computing}.

An important result from this space is from \citet{caragiannis2019unreasonable} showing that an \EFone{} and Pareto Optimal allocation is guaranteed to exist. In our setting we find that this is no longer the case. We use the leximin optimal allocation to provide a characterization for settings where \EFone{} and \PO{} can simultaneously be satisfied. In prior work, \citet{plaut2020almost} use a specific type of leximin optimal allocation to find $\mathrm{EFX}$ allocations under subadditive valuations. We also provide a characterization of when \EFone{} and socially optimal (SO) allocations exist. Such an allocation would maximize the social welfare (sum of all agents' values) and also be fair. Prior work has looked at maximizing SW over \EFone{} allocations \citep{barman2019fair,benabbou2020finding,caragiannis2022repeatedly}.

Submodular valuations and their subclasses have been well studied in prior work. Typically, submodular valuations require oracle access as the functions tend to be too large to be sent as an input to the algorithm. Depending on the type of oracle access, the abilities and efficiency of the algorithm changes. Subclasses of submodular valuations have also been explored in fair division literature, especially in the context of \MMS{} \citep{barman2021existence,barman2020approximation,ghodsi2018fair,benabbou2020finding,li2021estimating,barman2022truthful}. Additionally, while the majority of the work in this space has looked at settings with goods, some recent work also looks at chores, either alone or in conjunction with goods \citep{aziz2022fair,bhaskar2021approximate,caragiannis2022repeatedly,hosseini2022ordinalgoods,huang2016fair,freeman2019equitable}. 

Submodular costs have been studied in various algorithmic settings, be it combinatorial auctions, facility location or other  graph problems like shortest cycles. Of these the only study to look at a cost model similar to ours is \citet{svitkina2006facility} where the authors consider submodular facility costs using a rooted tree whose leaves are the facilities. The facility costs of opening a certain set of facilities was the sum of the weights of the  vertices needed to be crossed from the root to reach reach these facilities from the root. While this is very similar to the cost functions in our model, it is important to note that this work has no fairness considerations, only aiming to minimize the total costs. Needless to say, this can cause a large discrepancy in the workloads of different facilities.

Price of fairness has also been studied in both additive and submodular valuation settings \citep{barman2020optimal,bei2021price,caragiannis2012efficiency,sun2021fairness,bhaskar2023price,li2024complete}. Recently, \citet{li2024complete} provide a complete characterization of the price of envy-based fairness notions under divisible and indivisible goods. \citet{sun2021fairness} looks at the price of \EFone{} and approximate \MMS{} among other notions in chore division settings. They also show how the notions relate to the each other. 
In particular, they demonstrate that the price of fairness can be arbitrarily large.
This is not the case in our setting
where the price of fairness is bounded
for a given $m$ and $n$ (\cref{prop:priceefone}).

\subsection{Fair Task Allocation}

Task allocation is a bit of an overloaded term and has been used outside of fair allocations of chores. For the most part, it refers to efficiently allocating tasks to robots (see \cite{nunes2017taxonomy} for an overview) and is not concerned with fairness. When pursuing fair task allocations, the tasks considered are independent of each other, unlike our model where tasks lie on a graph. Most settings studied require  exactly one agent to complete a task \cite{ye2017fair,basik2018fair,zlotkin1992domain,sun2020fair}. Some others may require groups to be allocated \cite{amador2014dynamic} or allow tasks to be divided among multiple agents \cite{billing2020fair}.

Fairness here is pursued in different ways. For the most part, it focuses on either i) reducing variance in the assigned loads of agents \cite{basik2018fair,billing2020fair} or maximizing the minimum number of tasks assigned \cite{ye2017fair}. For the most part, no costs are considered here, with the exception of the costs of the agents commissioning the tasks \cite{ye2017fair}. Most works assume prior knowledge of upcoming tasks \cite{ye2017fair,basik2018fair,zlotkin1992domain,billing2020fair} and give algorithms that allocate these tasks in advance. \citet{amador2014dynamic} also looks at a setting where tasks arrive dynamically, but the tasks here are largely independent. Each task has an associated start and end time. As a result, an agent can only be allocated tasks whose start and end times do not overlap. In contrast, out model does not how such restrictions. 

\subsection{Fair Food Delivery}
Work on fairness in delivery settings has been almost entirely empirical \citep{nair2022gigs,gupta2022fairfoody,wu2022fair,singh2023fairassign,tsai2023genetic} and no positive theoretical guarantees have been provided. While the models pursued are not identical, one common assumption is that an agent services one request at a time. Further, the cost of servicing a set of requests is simply the sum of the costs of the individual requests. This clearly differs from our setting where agents service a large set of orders together and the cost of the set is a submodular function of the costs of the individual orders. To the best of our knowledge, no other paper has studied interdependent costs of delivery tasks.  

Additionally, under prior work on food delivery, the aim is typically to achieve fairness by income distribution\cite{nair2022gigs,gupta2022fairfoody,singh2023fairassign} . Specifically, the goal is reducing the variance in the income earned by the various delivery workers. In contrast,  we look for fairness in the workload of the delivery agents, as there can be many settings where agents receive no or fixed monetary compensation for their work. In our setting, fairness is achieved by a fair division of the underlying graph.

\subsection{Fairness with Graphs}
Some prior work has also looked at fair division on graphs \citep{bouveret2017fair,bilo2022almost,truszczynski2020maximin,misra2021equitable,misra2022fair}. With the exception of \citet{misra2022fair}, all other papers looked at settings where items are on a graph and each agent must receive a {\em connected} bundle.
Our model does not have such a restriction. The purpose behind the prior work done in fair division on graph is to partition the graph into $n$ vertex-disjoint connected subgraphs. The value of the agents comes only from the vertices in their bundle, and does not depend on vertices outside of their bundle. In our case, the subgraphs will always intersect in the hub. Additionally, for us, the bundle that an agent receives need not form a connected subgraph. The cost depends on the distances of these vertices to the hub, which is in not allocated. Thus, agents can incur costs from having to traverse vertices that are not in their bundles. 
\citet{misra2022fair} look at a setting where agents are connected by a social network and only envy those agents they have an edge to. 

Recent work, initiated by \citet{christodoulou2023fair} introduces {\em graphical valuations} where agents are vertices on a graph and the edges are the items \cite{bhaskar2024efx,zeng2024structure,misra2024envy,deligkas2024ef1,afshinmehr2024efx,hsu2024efx,hsu2025polynomial,zhou2024complete}. An incident edge can only be assigned to one of its two endpoints. It is easy to see that there is little overlap between this setting and the one introduced in this paper.

While fairness has been studied in routing problems, the aim has been to balance the amount of traffic on  each edge \citep{kleinberg1999fairness,pioro2007fair}. This does not capture the type of delivery instances that we look at. Some work has also looked at ride-hailing platforms and aiming to achieve group and individual fairness in these settings \citep{sanchez2022balancing,esmaeili2022rawlsian}. However, these studies are largely experimental and do not provide any theoretical guarantees. Further, these models also do not look at models with submodular costs. In fact, they use linear programs to achieve their experimental results. 

\section{Graph Preliminaries}\label{app:treeprelims}
A \emph{graph} is defined by a pair $G=(V,E)$,
where $V$ (or $V(G)$) is a set of \emph{ vertices}
and $E$ (or $E(G)$) is a set of (undirected) \emph{edges}.

We are often interested in some specific sub elements of a given graph. We shall now define those relevant to this paper.
A \emph{walk} is a sequence of  vertices $(v_0,\dots,v_k)$
such that every two consecutive  vertices are connected by an edge.
The \emph{length} of a walk is the number of edges in it,
so the number of  vertices minus 1.
A \emph{path} is a walk in which all  vertices are pairwise distinct.

A walk or a path between two vertices allows us to define several useful properties. We say that a graph is \emph{connected} if there exists a walk between every pair of vertices in $V$.
In a connected graph, the distance between vertex $u$ and $v$,
i.e., $\dist(u,v)$, is the minimum possible length
of a walk connecting $u$ and $v$.

We now discuss some graph notation. A \emph{subgraph} of graph $G$
is any graph $H=(U,M)$ such that
$U \subseteq V$ and $M \subseteq E$.
A subgraph is \emph{induced}
if for every edge $(u,v) \in E$
such that $u,v \in U$ we have $(u,v) \in M$.

A \emph{tree} is a graph in which
between every pair of  vertices
there exists exactly one path. Some properties of a tree $G=(V,E)$ with $n$ vertices are as follows:
\begin{itemize}
    \item $G$ is connected 
    \item $G$ is acyclic
    \item $G$ contains exactly $n-1$ edges. 
    \item Given a subset of vertices $S\subset V$, a minimum length walk containing all the vertices in $S$ can be found in polynomial time.
\end{itemize}

We now finally define some terms related to trees that enable convenient reference in our proofs. A \emph{rooted tree}, $(G,h)$, is a tree with one vertex, $h \in V$, designated as the \emph{root}. In our paper we shall assume that our trees are rooted at the hub.

For every vertex $u \in V$, every vertex in the path
from $u$ to $h$ except for $u$ is called its \emph{ancestor}.
For such  vertices $u$ is a \emph{descendant}. An ancestor (or descendant) that is connected to $u$ by an edge
is called a \emph{parent} (or \emph{child}).

A vertex without children is called a \emph{leaf} (vertex);
otherwise, if it does have children, it is called an \emph{internal} vertex.
A \emph{subtree rooted in} $u$ is
an induced subgraph $T_u = (U,M)$ rooted in $u$,
where $U$ contains $u$ and all its descendants.
By a \emph{branch} outgoing from $h$ we understand
a set of vertices in 
a subtree rooted in a child of $h$.


\section{Relation Between MMS and EF}\label{sec:mmsandfriends}
In this section,
we consider 
the relation between \MMS{} and envy-freeness.


It is well-known that \EF{} implies \MMS{} for additive items, via Proportional share. 
However, \EFone{} does not imply \MMS{}, even for identical additive valuations and vice versa. However, for identical additive settings an allocation that is simultaeneously \EFX{} and \MMS{} must exist. 


\begin{proposition}
    Given a $n$ identical agents and $m$ indivisible goods $M$ and additive valuation function $v:2^M\rightarrow \mathbb{R}^+$, there exists an allocation $A$ s.t. $A$ is \EFX{} and \MMS{}. However, every \EFX{} allocation is not \MMS{} and every \MMS{} allocation is not \EFone{}.
\end{proposition}

\begin{proof}
    Consider the leximin optimal allocation $A^*$ for the $n$ agents over $M$ under $v$. Clearly, $A^*$ is leximin optimal. We now show that $A^*$ must be \EFX{}. Without loss of generality let $v(A_1^*)\geq v(A_2^*)\geq \cdots \geq v(A_n^*)$. Further, we assume that for all $g\in M$, $v(g)>0$
    
    For contradiction, let $A^*$ not be \EFX{}. That is, there exists $i$ s.t. for some $g\in A_i$, $v(A_i^*\setminus g)>v(A_n^*)$. Fix  $g\in A^*_i$ s.t. $v(A_i^*\setminus g)>v(A_n^*)$. Consider the allocation $A$ where for $j\notin \{i,n\}$ $A_j=A^*_j$ and $A_i=A^*_i\setminus \{g\}$ and $A_n=A^*_n\cup \{g\}$. Now, by assumption $v(A_i)=v(A^*_i\setminus \{g\})>v(A_n^*)$ and clearly $v(A_n)>v(A^*_n)$. Thus, $A$ leximin dominates $A^*$. This contradicts the assumption that $A^*$ is leximin optimal.

    As a result, a leximin optimal allocation must be \MMS{} and \EFX{}.

    \paragraph{\EFX{} is not \MMS{}.} We show that \EFX{} does not imply \MMS{} for identical valuations (\EFX{} is a stricter notion than \EFone{} that requires that on the removal of \emph{any} item, envy-freeness must be satisfied).  Take a simple example. Let $n=2$ and $m=4$, where $v(g_1)=v(g_2)=2$ and $v(g_3)=v(g_4)=1$. 
    Here an \MMS{} allocation would give both agents a value of 3. Now, the allocation $A$ where $A_1=\{g_1,g_2\}$ and $A_2=\{g_3,g_4\}$ is \EFX{} (and hence, \EFone{}) but not \MMS{}.

    \paragraph{\MMS{} is not \EFone{}.} Consider an example with $3$ agents and $5$ indivisible goods where $v(g)=1$ for each good $g$. Here, for an allocation to be \MMS{}, we need that $v(A_i)\geq 1$ for each $i$. Consider the allocation $A$ where $|A_1|=3$ and $|A_2|=|A_3|=1$. Clearly this is \MMS{} but not \EFone{}. 
\end{proof}

An analogous proof also shows that for identical additive cost functions over indivisible chores, an allocation that satisfies \MMS{} and \EFone{} always exists. 

Unfortunately, even \EF{} need not imply \MMS{} under our setting. Consider the instance and allocation depicted in \cref{fig:efmms}. Both agents incur a combined cost of 3, but the \MMS{} cost is 2.


\begin{figure}[t]
    \centering 
   \begin{tikzpicture}
        \def\xs{1.2cm} 
        \def\ys{0.4cm} 
        \def\x{0cm} 
        \def\y{0cm} 
        \def\ls{\footnotesize} 
        
        \tikzset{
            node_blank/.style={circle,draw,minimum size=0.5cm,inner sep=0, color=white}, 
            node/.style={circle,draw,minimum size=0.45cm,inner sep=0, fill = black!05},
            node_h/.style={circle,draw,minimum size=0.6cm,inner sep=0, fill = blue!20, font=\footnotesize},
            node_emph/.style={circle, minimum size=0.75cm, black!15, fill = black!15, draw,inner sep=0, font=\footnotesize},
            edge/.style={draw,thick,sloped,-,above,font=\footnotesize},
            arrow/.style={draw, single arrow, minimum width = 0.9cm, minimum height=\y-6*\x+\s, fill=black!10},
            blank/.style={},
            alloc_a/.style={circle, draw, color=blue!30, minimum size=0.8cm, line width=2.5},
            alloc_b/.style={rectangle, draw, color=red!30, minimum size=0.7cm, line width=2.5},
        }
        
        \node[node_emph] (_) at (\x + 1*\xs, 0*\ys + \y) {};
        \node[node] (a) at (\x + -1*\xs - 0.1cm, 0*\ys + \y) {\ls $a$};
        \node[node] (b) at (\x + -0*\xs - 0.1cm, 0*\ys + \y) {\ls $b$};
        \node[node_h] (h) at (\x + 1*\xs, 0*\ys + \y) {\ls $h$};
        \node[node] (c) at (\x + 2*\xs + 0.1cm, 0*\ys + \y) {\ls $c$};
        \node[node] (d) at (\x + 3*\xs + 0.1cm, 0*\ys + \y) {\ls $d$};
        
        \node[alloc_a] (_) at (\x + -1*\xs - 0.1cm, 0*\ys + \y) {};
        \node[alloc_b] (_) at (\x + -0*\xs - 0.1cm, 0*\ys + \y) {};
        \node[alloc_a] (_) at (\x + 2*\xs + 0.1cm, 0*\ys + \y) {};
        \node[alloc_b] (_) at (\x + 3*\xs + 0.1cm, 0*\ys + \y) {};
        
        \path[edge]
        (a) edge (b)
        (b) edge (h)
        (h) edge (c)
        (c) edge (d)
        ;

    \end{tikzpicture}
    \caption{An example graph showing that \EF{} does not imply \MMS{}.
    The visible allocation, i.e., $(\{a,c\},\{b,d\})$, is \EF{}
    but not \MMS{}.}
    \label{fig:efmms}
\end{figure}


\section[Additional Material for Section 6]{Additional Material for Section 6}\label{app:xp}
In this section,
we describe \cref{alg:ef1+po:main} in more detail 
and formally prove its correctness.
We further show how it can be used to find
allocations satisfying a combination
of fairness and efficiency notions
whenever they exist.

We begin by introducing some additional notation.
By $\Pi_n$ let us denote the set of all permutations
of set $[n]$.
For a permutation $\pi \in \Pi_n$
and allocation $A = (A_1,\dots,A_n)$,
by $\pi(A)$ we understand allocation $(A_{\pi(1)},\dots,A_{\pi(n)})$,
i.e., allocation where agent $i$ receives the original bundle of agent $\pi(i)$.
For two partial allocations,
$A = (A_1,\dots,A_n)$ and $B=(B_1,\dots,B_n)$,
with disjoint set of distributed vertices,
i.e., $\bigcup_{i\in [n]} A_i \cap \bigcup_{i \in [n]}B_i = \emptyset$,
by $C = A \oplus B$,
we understand allocation $C = (A_1 \cup B_1,\dots,A_n \cup B_n)$.
For two allocations $A$ and $B$ we write $A <_{PO} B$,
when $A$ is Pareto dominated by $B$.
We write $A \le_{PO} B$, when it is \emph{weakly Pareto dominated},
i.e., $c(A_i) \ge c(B_i)$ for every $i \in [n]$.
Now, we are ready to describe \cref{alg:ef1+po:main} with details on
how we combine the two lists of allocations, $\mathcal{F}$ and $\mathcal{F}'$.

\setcounter{algorithm}{0}
\begin{algorithm}[!t]
    \caption{FindParetoFrontier($n$, $G$,$h$)}
    \label{alg:ef1+po:main_app}
    \begin{algorithmic}[1]
        \STATE $\mathcal{F} \leftarrow [ (\emptyset, \dots, \emptyset) ]$
        \FOR{$u \in$ children of $h$}
            \STATE $T_u \leftarrow$ a subtree rooted in $u$
            \STATE $\mathcal{F}' \leftarrow $ FindParetoFrontier($n$, $T_u$, $u$)
            \STATE \textbf{for} $A \in \mathcal{F}'$ \textbf{do} add $u$ to $A_1$
            \STATE $\mathcal{F}'' \leftarrow [\ ]$, empty list of allocations
            \FOR{ $A \in \mathcal{F}, B \in \mathcal{F}', \pi \in \Pi_n$}
                \STATE $C =$ sort$(A \oplus \pi(B))$
                \IF{there is no $D \in \mathcal{F}''$ s.t. $C \le_{PO} D$}
                    \STATE add $C$ to $\mathcal{F}''$ 
                    \WHILE{there is $D \in \mathcal{F}''$ s.t. $D <_{PO} C$}
                    \STATE remove $D$ from $\mathcal{F}''$
                    \ENDWHILE
                \ENDIF
            \ENDFOR
            \STATE $\mathcal{F} \leftarrow \mathcal{F}''$
        \ENDFOR
        \RETURN $\mathcal{F}$
        \end{algorithmic}
\end{algorithm}


\paragraph{Algorithm.}
Throughout the algorithm,
we keep allocations in the list $\mathcal{F}$,
each allocation sorted in non-increasing cost order.
First, we initialize it
with just one empty allocation.
Then, we look at vertices directly connected to the hub.
For each such $u$, we run our algorithm on a smaller instance
where the graph is just the branch outgoing from $h$
that $u$ is on, and $u$ is the hub.
In each allocation in the output, $\mathcal{F}'$,
we add $u$ to the bundle of the first agent.
Finally, we combine the allocations returned by each child of $h$.
%

To this end, we first initialize $\mathcal{F}''$ with an empty list of allocations.
We then iterate over all possible triples $(A,B,\pi)$,
where $A$ and $B$ are allocations
from input lists $\mathcal{F}$ and $\mathcal{F}'$, respectively,
and $\pi$ is a permutation of agents $[n]$.
For each such triple,
we consider allocation $C = \sort(A \oplus \pi(B))$, i.e.,
the allocations in which agent $i$,
receives bundle $A_i \cup B_{\pi(i)}$,
sorted in non-increasing cost order.

In the next step, we check if there exists
an allocation $D$ in $\mathcal{F}''$ that weakly Pareto dominates $C$.
If this is the case, we disregard $C$
and move to the next triple.
If this is not the case,
then we add allocation $C$ to the list $\mathcal{F}''$
and remove all allocations $D$ from $\mathcal{F}''$
that are Pareto dominated by $C$.
We note that these operations
can be performed more efficiently
if we keep allocations in $Res$ in a specific ordering,
but since it is not necessary for our results,
we do not go into details for the sake of simplicity.
After considering all pairs and permutations,
we put $\mathcal{F}''$ for $\mathcal{F}$.

\xpalgproof*

\begin{proof}
Let us start by showing that the output is a Pareto frontier,
which we will prove by induction on the number of edges in a graph.
If there is only one edge,
the graph consists of the hub, $h$,
and one vertex connected to it, say $u$.
Then, there is only one possible allocation
(up to a permutation of agents),
namely, $(\{u\},\emptyset,\dots,\emptyset)$,
and it is \PO.
Observe that this is also
the only allocation returned by our algorithm
for such a graph.
Hence, the inductive basis holds.

Now, assume that our algorithm outputs a Pareto frontier for every instance,
in which the number of edges is smaller or equal to $M$ for some $M \in \mathbb{N}$
and consider an instance $\langle [n], G, h \rangle$ with $M+1$ edges.
Take arbitrary \PO{} allocation $A$.
We will show that in the output of \cref{alg:ef1+po:main}
for this instance, $\mathcal{F}$, there exists allocation $B$ such that
$c(A_i) = c(B_{\pi(i)})$ for some permutation $\pi \in \Pi_n$.

If $h$ has only one child, $u$,
then observe that every agent that services some vertex in $A$
has to visit $u$.
Also, the agent that services $u$ services also other  vertices
(otherwise giving $u$ to an agent that services some other vertex
would be a Pareto improvement).
Hence, partial allocation $A'$
obtained from $A$ by removing $u$ is still \PO{}
and the cost of each agent is the same in both allocations.
Observe that $A'$ is also a \PO{} allocation in instance $\langle [n], T_u, u \rangle$.
Let $\mathcal{F}'$ be the output of \cref{alg:ef1+po:main}
for instance $\langle [n], T_u, u \rangle$.
Then, since $T_u$ has $M$ edges,
from inductive assumption we know that there exists $B' \in \mathcal{F}'$
and $\pi \in \Pi_n$ such that
$c(A'_i) = c(B'_{\pi(i)})$, for every $i \in [n]$.
Then, let $B$ be an allocation obtained from $B'$
by adding $u$ to the bundle of agent $1$.
Since we sort allocation so that consecutive agents have nonincreasing costs,
we know that $B'_1 \neq \emptyset$.
Hence, cost of each agent in $B$ is the same as in $B'$.
Hence, $c(A_i) = c(A'_i) = c(B'_{\pi(i)}) = c(B_{\pi(i)})$.
Since $B$ is in the output of the algorithm,
the induction thesis holds.

Now, assume that $h$ has more than one child.
Let us denote them as $u^1,\dots,u^k$
and by $U^1,\dots,U^k$ let us denote
the respective branches outgoing from $h$.
Let $A^1,\dots,A^k$ be partial allocations obtained from $A$
by restricting $A$ to one branch from $U^1,\dots,U^k$, respectively
(i.e., removing all vertices not in the branch from all of the bundles).
Observe that since $A$ is \PO{},
each allocation $A^1,\dots,A^k$ is also \PO{}
(otherwise a Pareto improvement in $A^j$ for some $j \in [k]$
would be a Pareto improvement also in $A$).
With the same reasoning as in the previous paragraph,
by line 6 of \cref{alg:ef1+po:main}
in the iteration of the loop for each child $u^j$,
list $\mathcal{F}'$ contains an allocation $B^j$ such that
$c(A^j_i) = c(B^j_{\pi^j(i)})$ for some $\pi^j \in \Pi_n$ and every $i \in [n]$.
Now, when we combine $\mathcal{F}$ with $\mathcal{F}'$
we consider all possible combinations of allocations in $\mathcal{F}$ and $\mathcal{F}'$
along with all possible permutations of agents.
Hence, in the output of the algorithm,
there will be allocation $B$ and permutation $\pi \in \Pi_n$ such that
$B_{\pi(i)} = B^1_{\pi^1(i)} \cup \dots \cup B^k_{\pi^k(i)}$
for every $i \in [n]$,
unless there is some allocation $D$ that weakly Pareto dominates $B$.
Since $A^1,\dots,A^k$ are partial allocations of separate branches
and $B^1,\dots,B^k$ as well, we have
\[
    c(B_{\pi(i)}) = 
        \textstyle \sum_{j \in [k]} c(B^j_{\pi^j(i)}) = 
        \textstyle \sum_{j \in [k]} c(A^j_i) = 
        c(A_i),
\]
for every $i \in [n]$.
Hence, $B$ is \PO.
Thus, if there is $D$ that weakly Pareto dominates $B$,
then $c(D_i)=c(B_i)$ for every $i \in [n]$.
Either way, there exists an allocation in $\mathcal{F}$
that for corresponding agents gives the same costs as allocation $A$.
Therefore, $\mathcal{F}$ is a Pareto frontier, which concludes the induction proof.

In the rest of the proof,
let us focus on showing that the running time of \cref{alg:ef1+po:main}
is $O((n+2)!m^{3n+1})$.
To this end, recall that
in the nested loop of the algorithm (lines 7-15)
we consider all triples $(A,B,\pi)$,
where $A$ is an allocation in $\mathcal{F}$,
$B$ an allocation in $\mathcal{F}'$,
and $\pi$ a permutation in $\Pi_n$.
Since the cost of an agent is an integer between $0$ and $m$
and at any moment we cannot have
two allocation with the same cost for every agent
in $\mathcal{F}$ or $\mathcal{F}'$
(because one is weakly Pareto dominated by the other),
the sizes of $\mathcal{F}$ and $\mathcal{F}'$ are bounded by $(m+1)^n$.
Hence, the nested loop in lines 7--15
will have at most $n!(m+1)^{2n}$ iterations.
For each triple, we have to sort the resulting allocation, which can take time $n\log(n)$.
Moreover, we may need to check whether resulting allocation $C$
Pareto dominates or is Pareto dominated by all of the allocations already kept in $\mathcal{F}''$.
The size of $\mathcal{F}''$ is also bounded by $(m+1)^n$
and checking Pareto domination can be done in time $n$.
All in all, the running time of
the nested loop is in $O((n+2)! m^{3n})$.
Finally, observe that in the algorithm 
we go through the nested loop less than $m$ times,
thus final running time is in $O((n+2)! m^{3n + 1})$.
\end{proof}

Now, let us show how we can use \cref{alg:ef1+po:main}
to find allocation satisfying certain fairness and efficiency requirements
(or decide if they exist).

\xpalgallocproof*
\begin{proof}
    Let us split the proof into four lemmas devoted
    to each combination of fairness and efficiency notions.
    We start with \MMS{} and \PO{} allocations.

    \begin{lemma}
        \label{thrm:mms+po:xp}
        There exists an XP algorithm parameterized by $n$ that
        for every delivery instance $\langle [n], G, h \rangle$
        finds an \MMS{} and \PO{} allocation.
    \end{lemma}
    \begin{proof}
        Observe that the allocation in a Pareto frontier
        that leximin dominates all other allocations in the frontier
        is a leximin optimal allocation.
        Since \cref{alg:ef1+po:main} returns a Pareto frontier,
        and Pareto frontier contains at most $O(m^n)$ allocations,
        we can find such an allocation in XP time by \cref{thm:xpalg}.
        Therefore, the lemma follows from 
        the fact that leximin optimal allocation
        is \MMS{} an \PO{}.
    \end{proof}

    Next, let us move to \EFone{} and \PO{} allocations.
    \begin{lemma}
    \label{thrm:ef1+po:xp}
        There exists an XP algorithm parameterized by $n$ that
        for every delivery instance $\langle [n], G, h \rangle$
        decides if there exists an \EFone{} and \PO{} allocation
        and finds it if it exists.
    \end{lemma}
    \begin{proof}
        From the proof of \cref{thrm:ef1+po:implies_mms} and \cref{lem:ef1+po}
        we know that an \EFone{} and \PO{} allocation is
        leximin optimal and
        the pairwise differences in the costs of agents are at most 1.
        Observe that it is an equivalence, i.e.,
        a leximin optimal allocation in which
        the pairwise differences in the costs of agents are at most 1
        is \EFone{} and \PO{}
        (\PO{} because of leximin optimality, and \EFone{} because for every agent
        there exists a vertex that removed from a bundle of this agent
        reduces the cost by at least 1).
        By \cref{thrm:mms+po:xp}, we can find a leximin optimal
        allocation in XP time with respect to $n$.
        Therefore, it remains to check if
        the pairwise differences in the costs of agents are at most 1.
        If it is true, then we know that this allocation is \EFone{} and \PO.
        Otherwise, we know there is no \EFone{} and \PO{} allocation.
    \end{proof}

    Now, let us consider \MMS{} and \SO{} allocations.

    \begin{lemma}
        \label{thrm:so:xp}
        There exists an XP algorithm parameterized by $n$ that
        for every delivery instance $\langle [n], G, h \rangle$
        decides if there exists an \MMS{} and \SO{} allocation
        and finds it if it exists.
    \end{lemma}
    \begin{proof}
        For \MMS{} and \SO{},
        observe that an \SO{} allocation is also necessarily a \PO{} allocation.
        Having Pareto frontier from \cref{thm:xpalg}
        and \MMS{} value from \cref{thrm:mms+po:xp},
        we can simply check all allocation in the frontier
        if they satisfy \MMS{} and \SO.
    \end{proof}

    Finally, we focus on \EFone{} and \SO{} allocations.
    \begin{lemma}
        \label{thrm:ef1+so:xp}
        There exists an XP algorithm parameterized by $n$ that
        for every delivery instance $\langle [n], G, h \rangle$
        decides if there exists an \EFone{} and \SO{} allocation
        and finds it if it exists.
    \end{lemma}
    \begin{proof}
        For \EFone{} and \SO,
        observe that \cref{thrm:ef1+po:implies_mms} also implies that 
        \EFone{} and \SO{} allocation is \MMS.
        From \cref{thrm:so:xp} we know that we can find all
        \MMS{} and \SO{} allocations in Pareto frontier
        in XP time with respect to $n$.
        Hence, we can check if any one of them satisfies also \EFone{}
        and this will give us the solution.
    \end{proof}

    From \cref{thrm:mms+po:xp,thrm:ef1+po:xp,thrm:so:xp,thrm:ef1+so:xp}
    we obtain the thesis.
\end{proof}


\section{Extensions to Weighted/Cyclic Graphs and Future Work}\label{sec:extend}


This work is the first to theoretically consider standard fairness objectives for delivery tasks.
Trees, as we discuss in the introduction are an important setting to consider, not the least because they allow for tractable routing solutions. For this reason, we explore and exploit the structure of trees in order to get our results. 


It is important to note that several of our results extend to weighted cyclic graphs as well. In particular, all our hardness results and the existence of \EFone{} allocations extend.  
Further, our hardness results show that when there are edge weights, looking for \MMS{} allocations or allocations that satisfy any combination of fairness and efficiency is \emph{strongly} \NPH{}. Thus our results have strong implications for delivery settings with weights and/or cycles in the graph. We now remark on each of these settings separately and remark on the avenues for future work.

\subsection{Weighted Trees}

The setting of weighted trees subsumes the standard fair chore division setting with identical additive costs. Any such chore division setting can be captured by a weighted star graph with a separate hub and a vertex for every chore, and the edge between the hub and the vertex has weight equal to the cost of the chore. Consequently, the setting is significantly more difficult than typical chore division settings.

As previously mentioned, all the results in \cref{sec:fairallocations} extend to weighted trees. Here, oracle access is not explicitly required to find \EFone{} allocations, unlike the case of cyclic graphs. Further, \cref{alg:ef1+po:main} can easily extend to weighted trees, however
the running time of such algorithm
will have the sum of the edge weights
across the tree in place of $m$.

\begin{proposition}
    Given a delivery instance with a weighted tree, an \EFone{} allocation can be found in polynomial time. 
\end{proposition}

An important direction of further study would be to find a faster algorithm to find a fair and efficient algorithm, if possible. Also, studying the connection between \EFone{} and \MMS{} in this space would be very interesting.

Additionally, for the setting of unweighted trees itself, our work paves the way for future research on developing approximation schemes
or perhaps algorithms parameterized by graph characteristics
(e.g., maximum degree or diameter) in this domain.

\subsection{Cyclic Graphs}
While cyclic graphs are practically relevant, they bring intrinsic computational hardness. 
In fact, even the cost of servicing a given set of nodes in a weighted cyclic graph cannot be computed tractably.

\paragraph{Cyclic Weighted Graphs} With edge weights, the fair delivery problem generalizes the travelling salesperson problem, which is hard to approximate within \emph{any} polynomial factor. Therefore, while the envy-cycle elimination argument does guarantee the existence of \EFone{} allocations, to compute one would require access to a value oracle which will provide the cost of servicing a given set of nodes. 

The remainder of our results rely on the graph being acyclic, hence do not extend to cyclic graphs. An important avenue of future work includes understanding if there is any connection between \EFone{} and \MMS{} allocation. Further, it is relevant to note that an efficient way of verifying if a given allocation is \PO{} may not exist in that setting (it is co-NP complete for standard settings). Hence, even a brute force approach to finding a fair and efficient allocation for cyclic graphs may take time $\Omega(m^{n})$.

\end{document}